\documentclass[12pt]{iopart}

\usepackage{iopams}
\usepackage{amsthm}        
\usepackage{graphics,epsfig,psfrag}
\usepackage{color}
\usepackage{harvard}
\usepackage{bm}

\theoremstyle{definition}\newtheorem{definition}{Definition}
\theoremstyle{definition}\newtheorem{proposition}{Proposition} 

\def\cRed#1{#1}

\def\Man{{\cRed{\cal M}}}

\begin{document}

\title[]{On the self-force in Bopp-Podolsky electrodynamics}
\author{Jonathan Gratus$^1$\footnote{Email: j.gratus@lancaster.ac.uk} { , } 
Volker Perlick$^2$\footnote{Email: perlick@zarm.uni-bremen.de} { and }
Robin W. Tucker$^1$\footnote{Email: r.tucker@lancaster.ac.uk}
}

\address{$^1$Physics Department, Lancaster 
University, Lancaster LA1 4YB, UK; and The Cockcroft Institute, 
Warrington WA4 4AD, UK}
\address{$^2$ZARM, University of Bremen, 28359 Bremen, Germany}

\begin{abstract}
  In the classical vacuum \cRed{Maxwell-Lorentz} theory the self-force
  of a charged point particle is infinite. This makes \cRed{classical}
  mass renormalization necessary and, \cRed{in the special
    relativistic domain,} leads to the Abraham-Lorentz-Dirac equation
  of motion possessing unphysical run-away and pre-acceleration
  solutions. In this paper we investigate whether the higher-order
  modification of classical vacuum electrodynamics suggested by Bopp,
  Land\'{e}, Thomas and Podolsky in the 1940s, can provide a solution
  to this problem. Since the theory is linear, Green-function
  techniques enable one to write the field of a charged point particle
  on Minkowski spacetime as an integral over the particle's history.
  By introducing the notion of timelike worldlines that are ``bounded
  away from the backward light-cone'' we are able to prescribe
  criteria for the convergence of such integrals.  \cRed{We} also
  exhibit a timelike worldline yielding singular fields on a lightlike
  hyperplane in spacetime. In this case the field is mildly singular
  at the event where the particle crosses the hyperplane.  { Even in
    the case when the Bopp­Podolsky field is bounded, it exhibits a
    directional discontinuity as one approaches the point particle.
    We describe a procedure for assigning a value to the field on the
    particle worldline which enables one to define a finite Lorentz
    self-force.}  \cRed{This is explicitly derived leading to an
    integro-differential equation for the motion of the particle in an
    external electromagnetic field. We conclude that any worldline
    solutions to this equation belonging to the categories discussed
    in the paper have continuous 4-velocities.}
\end{abstract}

\vspace{2pc}


\section{Introduction}\label{sec:intro}

For many applications it is reasonable to model
moving charges in terms of classical charged
point particles. In accelerator physics, for example,
it is usually neither desirable nor feasible to 
model particle beams in terms of extended classical
charged bodies or of quantum matter.
Therefore a mathematically consistent theory of
classical charged point particles is of high
relevance.

Unfortunately, such a theory does not exist so far. Of course, there
is no problem as long as we restrict to a classical charged
\emph{test} particle \cRed{and} neglect the particle's \cRed{self-interaction}. Then
the equation of motion is just the \cRed{relativistic generalization
  of Newton's equation of motion with the rate of change of particle
  momentum equated to the Minkowski (relativistic) Lorentz force} in a
given external field and everything is fine. If, however, the
self-field is taken into account, the theory becomes
pathological. According to the \cRed{Maxwell-Lorentz theory} in vacuo,
the electromagnetic field of a point charge becomes infinite at the
position of the charge, so the particle experiences an infinite
self-force.  This infinity is so bad that the field energy in an
arbitrarily small ball around the source is infinite which leads to an
infinite term in the equation of motion of a point charge.
\citeasnoun{Dirac1938} suggested to counter-balance this infinity by
postulating that the point charge carries a negative infinite ``bare
mass'' which leads to the Abraham-Lorentz-Dirac equation.  Even if one
accepts this ad-hoc idea of an infinite bare mass, the problem has not
been solved. The Abraham-Lorentz-Dirac equation is known to possess
unphysical behavior such as run-away solutions and
pre-acceleration. \cRed{Reviews} of this dilemma, including detailed
accounts of the history, \cRed{can be found in the comprehensive
  monographs by} \citeasnoun{Rohrlich2007} and \citeasnoun{Spohn2007}.

Since the \cRed{Maxwell-Lorentz theory with point charges} in vacuo
does not lead to a consistent equation of motion of charged point
particles, one might think about modifying this theory. In the course
of history at least two such modifications have been suggested, both
motivated by the desire of solving the problem of an infinite
self-force, namely the Born-Infeld theory and the Bopp-Podolsky
theory. The Born-Infeld theory is by far the better known of the
two. This theory, which was suggested by \citeasnoun{BornInfeld1934},
modifies the \cRed{source free} Maxwell vacuum theory by introducing
non-linearities containing a new hypothetical constant of Nature $b$
with the dimension of a (magnetic) field strength.  For $b \to \infty$
the \cRed{Maxwell equations in vacuo are} recovered. The fact that the
Maxwell theory is very well verified by many experiments is in
agreement with the Born-Infeld theory as long as $b$ is sufficiently
large. By contrast, the Bopp-Podolsky theory retains linearity but
introduces higher-derivative terms proportional to a factor $\ell ^2$
where $\ell$ is a new hypothetical constant of Nature with the
dimension of a length.  Again, for $\ell \to 0$ the \cRed{Maxwell-Lorentz
  theory} is recovered. The Bopp-Podolsky theory
was first suggested by \citeasnoun{Bopp1940}.  It was independently
rediscovered by \citeasnoun{Podolsky1942}. Both Bopp and Podolsky
formulated their theory in terms of an action functional and then
derived the field equation which is of fourth order in the
electromagnetic potential. As noted by both Bopp and Podolsky, this
fourth-order equation is equivalent to a pair of second-order
equations in a certain gauge. If rewritten in this form, the
Bopp-Podolsky field system coincides with those of a theory suggested
by \citeasnoun{LandeThomas1941}.  Similar to the Born-Infeld theory,
the Bopp-Podolsky theory was first formulated as a classical field
theory but with the intention of deriving a quantum version later.  In
particular, Podolsky pursued both the classical and the quantum
aspects of the theory in several follow-up articles with different
co-authors, see \citeasnoun{PodolskyKikuchi1944},
\citeasnoun{PodolskyKikuchi1945} and \citeasnoun{PodolskySchwed1948}.
In the present article we are interested only in the classical theory.

In both the Born-Infeld theory and the Bopp-Podolsky theory the
self-field is \cRed{bounded} for a \emph{static} point charge, i.e., for a
point charge that is at rest in some inertial system in Minkowski
spacetime. This was shown already in the earliest articles on these
theories. Moreover, in both theories for such a charge the field
energy in a ball of radius $R$ around the charge is finite, even in
the limit $R \to \infty$.  To the best of our knowledge in the
Born-Infeld theory little is known regarding such finiteness for
accelerated point charges.  In the Bopp-Podolsky theory the only
result about accelerated point charges that we are aware of is due to
\citeasnoun{Zayats2013}, who showed that the self-force is finite for a
uniformly accelerated particle on Minkowski spacetime.

It is the purpose of this paper to add some results
on the finiteness of the self-force in the 
Bopp-Podolsky theory. In our view, these results 
give strong support to the idea that the Bopp-Podolsky
theory provides a consistent theory of 
classical charged point particles including the 
self-force. In Section \ref{sec:bp} we
briefly review the basic field equations of the Bopp-Podolsky
theory on Minkowski spacetime, emphasizing the fact
that because of their linearity Green-function techniques
can be used. In Section \ref{sec:pointoff} we restrict 
the equations of Section \ref{sec:bp}  to the case \cRed{where}
the source of the electromagnetic field is a point charge
with a prescribed worldline on Minkowski spacetime. We discuss
various ways of writing the field strength at a point 
\emph{off the worldline} as an integral over the particle's 
history. As a first example, we treat the simple case of a 
point charge that is at rest in an inertial system.
In Section \ref{sec:pointon} \cRed{the general problem of assigning a
  value of the field strength {\em on the worldline} is
  discussed. This is a precursor to the formulation of an
  integro-differential equation for the motion of a point charge that
  accommodates its finite self-force in an external electromagnetic field.}
As a second example, we treat
the case of a uniformly accelerated charge. In 
Section \ref{sec:finite} we present our main results
on the finiteness of the field and of the self-force \cRed{for general
 motion}. 
We show that the self-force is finite unless the 
worldline approaches the light-cone in the past in a
very contrived manner. As  a third example, we discuss
such a pathological worldline where the self-field
actually diverges on a lightlike hyperplane and the
self-force diverges at one point on the worldline.
However, we demonstrate that even in this case the 
singularity of the field is so mild that it does
not cause a problem for the equation of motion.     
In Section \ref{sec:LD}  we discuss 
how the Abraham-Lorentz-Dirac equation comes about in
a particular limit as $\ell\to 0$, after classical mass renormalization.

In the body of the paper we formulate the Bopp-Podolsky theory on
Minkowski spacetime in an inertial coordinate system. However, we have
added an appendix where we consider the Bopp-Podolsky theory on a
\emph{curved} spacetime in arbitrary coordinates. This allows us to
derive the dynamical (Hilbert) electromagnetic stress-energy-momentum
tensor of the theory using exterior calculus.  The appendix also
includes a derivation of the \cRed{relativistic} Lorentz force
\cRed{in the Bopp-Podolsky theory} from this tensor which is crucial
for our reasoning in the body of the paper.

\section{Bopp-Podolsky theory}\label{sec:bp}

\vspace{0.1cm}

We consider \cRed{a time and space oriented} Minkowski spacetime with
standard inertial coordinates $\boldsymbol{x}=(x^0,x^1,x^2,x^3)$,
\cRed{with metric tensor}
\begin{equation}\label{eq:eta}
g = 
\eta _{ab} dx^a \otimes dx^b
\end{equation}
where $\big( \eta _{ab} \big) = \mathrm{diag} (-1,1,1,1)$.
\cRed{In this article all tensor field components on Minkowski
  spacetime are with respect to the class of global parallel bases
  adapted to these coordinates. Members of this class are related by
  elements of the proper Lorentz group, $SO(3,1)$.} Here and in the following, Einstein's summation convention 
is used for latin indices which take values 
$0,1,2,3$ and for greek indices which take values $1,2,3$. 
Latin indices are lowered and raised with $\eta _{ab}$
and with its inverse $\eta ^{ab}$, respectively. We use units
in which the \cRed{the speed of light} $c$ equal to 1.

The higher-order electrodynamics suggested by \citeasnoun{Bopp1940}
and, independently, by \citeasnoun{Podolsky1942} is based on 
the gauge invariant action functional
\begin{equation}\label{eq:Lag}
S[A]  \, = \, 
\int_{\Man} \left( \frac{1}{16 \pi} \, F^{ab}F_{ab} \, + \,
\frac{\ell ^2}{16 \pi} \, \partial ^c F^{ab} \partial _c F_{ab}
\, - \, A_a j^a \right) d^4x \, 
\end{equation}
\cRed{where $\Man$ is some compact region of Minkowski spacetime
  yielding a finite $S[A]$.
}
Here $A_a$ is the electromagnetic potential
\begin{equation}\label{eq:FA}
F_{ab} \, = \, \partial _aA_b \, - \, \partial _b A_a
\end{equation}
is the electromagnetic field strength, $j^a$ is a conserved current
density 4-vector field,  \cRed{$\{\partial_a=\partial/\partial x^a\}$} and $\ell$ is a hypothetical new constant 
of nature with the dimension of a length. \cRed{It proves expedient to
  derive the Bopp-Podolsky field equations in terms of smooth fields
  and a smooth current source on $\Man$ and then discuss particular
  singular solutions associated with a source having support on a
  timelike worldline. This eliminates the need to perform variations
  of (\ref{eq:Lag}) in a distributional context.} 
Note that\cRed{, for deriving the field equations by variational
  methods,} the action 
functional (\ref{eq:Lag}) can be equivalently replaced with
\begin{equation}\label{eq:Lag2}
\tilde{S}[A]  \, = \, 
\int_{\Man} \left( \frac{1}{16 \pi} \, F^{ab}F_{ab} \, + \,
\frac{\ell ^2}{8 \pi} \, \partial _a F^{ab} \partial ^c F_{cb}
\, - \, A_a j^a \right) d^4x  
\end{equation}
because the integrands differ only by a total divergence. 

The field equations of the Bopp-Podolsky theory result from 
varying the action functional (\ref{eq:Lag}) or (\ref{eq:Lag2})
with respect to the potential. They read
\begin{equation}\label{eq:bpF}
\partial ^b F_{ba}  - \ell ^2 \square \partial ^b F_{ba}  
\, = \, - \, 4 \, \pi \, j_a 
\end{equation}
or, in terms of the potential \cRed{in the  Lorenz gauge
$\partial^b A_b \, = \, 0$,}
\begin{equation}\label{eq:bp}
\square A_a - \ell ^2 \square ^2 A_a  \, = 
\, - \, 4 \, \pi \, j_a 
\; 
\end{equation}
where $\square  \, = \, \partial ^b  \partial_b$
is the wave operator. 
Field equations involving the operator 
$\square ^2$ have also been investigated by \citeasnoun{PaisUhlenbeck1950}, 
cf. \citeasnoun{Pavlopoulos1967}.

Both Bopp and Podolsky observed that the fourth-order differential
equation (\ref{eq:bp}) can be reduced to a pair of second-order
differential equations. More precisely, (\ref{eq:bp}) is equivalent
to 
\begin{equation}\label{eq:bphat}
\square \hat{A}{}_a  \, = \, - \, 4 \, \pi \, j_a 
\; , \qquad \partial _a \hat{A}{}^a \, = \, 0 \, ,
\end{equation}
\begin{equation}\label{eq:bptilde} 
\square \tilde{A}{}_a \, - \, \ell ^{-2} \tilde{A}{}_a
\, = \, - \, 4 \, \pi \, j_a 
\; , \qquad \partial _a \tilde{A}{}^a \, = \, 0 \, .
\end{equation}
This can be demonstrated in the following way. Assume we have a
solution $A_a$ to (\ref{eq:bp}). Then we define
\begin{equation}\label{eq:Adef} 
\hat{A}{}_a \, = \, A_a \, - \, \ell ^2 \square A_a
\; , \qquad 
\tilde{A}{}_a \, = \, - \, \ell ^2 \square A_a
\end{equation}
and it is readily verified that (\ref{eq:bphat}) and (\ref{eq:bptilde})
are indeed true. Conversely, assume that we have solutions
$\hat{A}{}_a$ and $\tilde{A}{}_a$ to (\ref{eq:bphat}) and (\ref{eq:bptilde}),
respectively. Then we define
\begin{equation}\label{eq:Acon} 
A_a \, = \, \hat{A}{}_a \, - \, \tilde{A}{}_a
\end{equation}
and it is readily verified that (\ref{eq:bp}) is true. This gives a one-to-one
relation between solutions to (\ref{eq:bp}) and pairs of solutions
to (\ref{eq:bphat}) and (\ref{eq:bptilde}) which allows \cRed{one} to view the 
Bopp-Podolsky theory as equivalent to a theory based on the two
equations (\ref{eq:bphat}) and (\ref{eq:bptilde}). The latter was 
suggested, shortly after Bopp but independently of him and shortly before
Podolsky, by \citeasnoun{LandeThomas1941}. In a quantized
version of the Land{\'e}-Thomas theory, (\ref{eq:bphat}) describes the 
usual (massless) photon while (\ref{eq:bptilde}) describes a hypothetical
``massive photon'' whose Compton wave length is equal to the new constant of
nature $\ell$. 

\cRed{To find} the retarded solution to the fourth-order
Bopp-Podolsky equations (\ref{eq:bp}), for any given divergence-free 
source $j_a$, we use the reduction to the second-order equations
(\ref{eq:bphat}) and (\ref{eq:bptilde}) and then apply standard
Green-function techniques. The details were worked out \cRed{by}
Land{\'e} and Thomas, see Section 9 in \cite{LandeThomas1941}.  The
retarded solution to
\begin{equation}\label{eq:Ghat} 
\square \hat{G} (\boldsymbol{x}-\boldsymbol{y})
\, = \, - \, \delta (\boldsymbol{x}-\boldsymbol{y})
\end{equation}
is
\begin{equation}\label{eq:Greenhat} 
\hat{G} (\boldsymbol{x}-\boldsymbol{y})
=
\left\{
\begin{array}{ll}
(2 \pi )^{-1} \delta \big( D(\boldsymbol{x}-\boldsymbol{y})^2 \big) 
\qquad & \textup{if} \ \boldsymbol{y} < \boldsymbol{x} \, , 
\\[0.4cm]
\qquad 0 & \textup{otherwise} \, ,
\end{array}
\right. 
\end{equation}
and the retarded solution to 
\begin{equation}\label{eq:Gtilde} 
\square \tilde{G} (\boldsymbol{x}-\boldsymbol{y}) 
\, - \, \ell ^{-2} \tilde{G} (\boldsymbol{x}-\boldsymbol{y})
\, = \, - \, \delta (\boldsymbol{x}-\boldsymbol{y})
\end{equation}
is
\begin{equation}\label{eq:Greentilde} 
\tilde{G} (\boldsymbol{x}-\boldsymbol{y})
\, = \, \left\{
\begin{array}{ll}
\displaystyle\!\!
(2 \pi )^{-1} \delta \big( D(\boldsymbol{x}-\boldsymbol{y})^2 \big)
\, - \, 
\frac{
J_1 \big( \ell^{-1} D(\boldsymbol{x}-\boldsymbol{y}) \big)
}{
4 \pi \ell D(\boldsymbol{x}-\boldsymbol{y})}
& \textup{if} \: \, \boldsymbol{y} < \boldsymbol{x} \, ,
\\[0.4cm]
\qquad  0 &\textup{otherwise} \, .
\end{array}
\right. 
\end{equation}
Here and in the following, $J_n$ is the $n$th order Bessel function
of the first kind, $\boldsymbol{y}<\boldsymbol{x}$ means that 
$\boldsymbol{y}$ is in the chronological past of $\boldsymbol{x}$ and 
$D(\boldsymbol{x}-\boldsymbol{y})$ is the Lorentzian distance between 
these two events,
\begin{equation}\label{eq:defD} 
D ( \boldsymbol{x}-\boldsymbol{y}) 
\, = \, \sqrt{- \eta_{ab}(x^a-y^a)(x^b-y^b)} \; .
\end{equation}
Hence, the retarded solution to (\ref{eq:bphat}) is
\begin{equation}\label{eq:Ahatret} 
\hat{A}{}_a(x) \, = \, 4 \pi 
\int _{\mathbb{R}^4} \hat{G}(\boldsymbol{x}-\boldsymbol{y})
\, j_a (\boldsymbol{y}) \, d^4\boldsymbol{y} \, , 
\end{equation}
the retarded solution to (\ref{eq:bptilde}) is
\begin{equation}\label{eq:Atilderet} 
\tilde{A}{}_a(\boldsymbol{x}) \, = \, 4 \pi 
\int _{\mathbb{R}^4} \tilde{G}(\boldsymbol{x}-\boldsymbol{y})
\, j_a (\boldsymbol{y}) \, d^4\boldsymbol{y} \, , 
\end{equation}
and the retarded solution to (\ref{eq:bp}) is
\begin{eqnarray}
A_a(\boldsymbol{x}) & = & 4 \pi \! \!
\int _{\mathbb{R}^4} \! \!
\Big( \hat{G}(\boldsymbol{x}-\boldsymbol{y})
-
\tilde{G}(\boldsymbol{x}-\boldsymbol{y}) \Big)
\, j_a (\boldsymbol{y}) \, d^4\boldsymbol{y} 
\nonumber
\\\label{eq:Aret} 
&= &
\int _{\boldsymbol{y}<\boldsymbol{x}} \! \! \!
\frac{J_1 
\big( \ell ^{-1} D(\boldsymbol{x}-\boldsymbol{y}) \big)
}{\ell D(\boldsymbol{x}-\boldsymbol{y})} \,
j_a ( \boldsymbol{y}) \, d^4 \boldsymbol{y} \, .
\end{eqnarray}
The Lorenz gauge condition is satisfied by $\hat{A}{}_a$, 
$\tilde{A}{}_a$ and $A_a $ when the current density satisfies 
the continuity equation
\begin{equation}\label{eq:cont} 
\partial _a j^a  \, = \, 0 \, .
\end{equation}
\cRed{Mathematically} $\hat{A}{}_a $ is the retarded potential of the
standard Maxwell theory.  For $\ell \to 0$ we have $\tilde{G} \to 0$
and hence $\tilde{A}{}_a \to 0$, so in this limit the standard Maxwell
theory is recovered, as is obvious from (\ref{eq:Lag}).

\cRed{Equation} (\ref{eq:Aret}) can be viewed as a map that assigns to
each current density $\boldsymbol{j}$ the corresponding retarded
Bopp-Podolsky potential $\boldsymbol{A}$. A general framework for
investigating the question of whether this map is well-defined would
be to assume that $\boldsymbol{j}$ is a (tempered) distribution and to
ask if $\boldsymbol{A}$ is again a (tempered) distribution. As the
Green function $\hat{G}-\tilde{G}$ does not satisfy a fall-off
condition in all spacetime directions, this is a non-trivial
question. In this paper we restrict to a more specific question.  We
assume that $\boldsymbol{j}$ is the current density associated with a
point charge, and investigate \cRed{whether} the integral on the right-hand
side of (\ref{eq:Aret}) converges for events $\boldsymbol{x}$ in the
\cRed{future} of the worldline of the charge. If this is the
case, the left-hand side of (\ref{eq:Aret}) is, of course, well
defined, not only as a distribution but even as a function.

Up to now we have discussed only how the electromagnetic field can be
calculated from its source, i.e., from the current that generates
the field. We also need the equation for the Minkowski Lorentz force
density which  arises from the divergence of \cRed{an electromagnetic} 
stress-energy-momentum tensor $T_{ab}$. 
In inertial coordinates on Minkowski spacetime, it reads
\begin{equation}\label{eq:Lorentz}
\partial ^c T_{ca} = F_{ab}\,j^b \, .
\end{equation}
This equation \cRed{defines} the \cRed{Minkowski} force density that
the field $F_{ab}$ exerts on the current $j^b$. In the case that the
current is concentrated on a worldline, it gives the self-force
\cRed{in terms of field components on the worldline}. A
derivation of (\ref{eq:Lorentz}) in terms of the
stress-energy-momentum tensor $T_{ab}$ associated with the
Bopp-Podolsky theory, is given in the Appendix.

\section{The field of a point charge off the worldline}\label{sec:pointoff}

\vspace{0.2cm}

Let $\boldsymbol{\xi} (\tau)= \big( \xi ^0 (\tau), \xi ^1 (\tau), 
\xi ^2 (\tau), \xi ^3 (\tau)\big)$ be an inextendible timelike 
$C^{\infty}$ curve parametrized
by proper time,
\begin{equation}\label{eq:tau} 
\dot{\xi}{}^a ( \tau ) \, \dot{\xi}{}_a ( \tau) \, = \, - \, 1 \, .
\end{equation}
\cRed{
  Inextendible curves have parametrizations with $\tau$ belonging to an
  open interval  $\; ]\,\tau _{\mathrm{min}} , \tau _{\mathrm{max}} \,
  [\;$ where $\tau _{\mathrm{min}}\in\mathbb{R}\cup\{-\infty\}$ and 
  $\tau _{\mathrm{max}}\in\mathbb{R}\cup\{+\infty\}$.}\footnote{
\cRed{Some curves with unbounded acceleration reach past or future
infinity in a finite proper time.}}


\begin{figure}
\psfrag{xi}{\makebox(0,0)[r]{$\boldsymbol{\xi}(\tau)$}}
\centerline{\includegraphics[width=6cm]{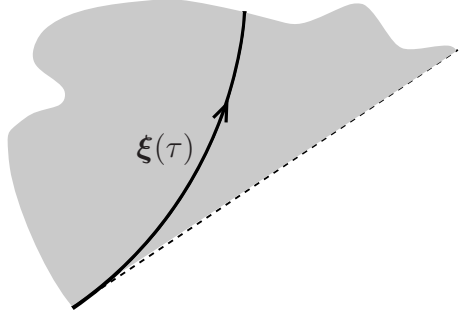}}
\caption{\cRed{When the future pointing timelike worldline
    $\boldsymbol{\xi}(\tau)$ approaches a lightlike line (shown dashed)
    as $\tau\to\tau_{\min}$, its future is the subset of
    Minkowski spacetime, partially shown in grey above, bounded by a
    lightlike hyperplane containing the dashed line. The grey domain only
    consists of  events which can be reached from the worldline along
a future-oriented timelike curve.}
}
\label{fig:future}
\end{figure}

Consider the \cRed{future} of \cRed{$\boldsymbol{\xi}(\tau)$}, i.e., 
the set of all events that can be reached from the worldline along a 
future-oriented timelike curve. If the worldline approaches a
light-cone asymptotically for $\tau \to \tau _{\mathrm{min}}$, 
its future is bounded by a lightlike hyperplane, 
see Fig.~\ref{fig:future}; otherwise it is all of $\mathbb{R}^4$. 
To each $\boldsymbol{x}$, in the future of
$\boldsymbol{\xi}(\tau)$, but not on the worldline,
we assign the retarded time $\tau _R 
( \boldsymbol{x} )$, defined by the properties that
\begin{equation}\label{eq:tauR}
\Big( x^a - \xi ^a \big( \tau _R ( \boldsymbol{x} ) \big) \Big)  
\Big( x_a - \xi _a \big( \tau _R ( \boldsymbol{x} ) \big) \Big) = 0
\, , \qquad
x^0 > \xi ^0 \big( \tau _R ( \boldsymbol{x} ) \big) \, , 
\end{equation}
and introduce the retarded distance 
\begin{equation}\label{eq:rR}
r_R ( \boldsymbol{x} ) \, = \, - \, 
\dot{\xi}{}^a \big( \tau _R ( \boldsymbol{x} ) \big) 
\Big( x_a - \xi _a \big( \tau _R ( \boldsymbol{x} ) \big) \Big) \; .
\end{equation}
\cRed{Unless otherwise specified, in the following, all events
$\boldsymbol{x}$ will belong to the future of the
worldline. However, whether such events lie on or off the worldline
will be made explicit.}
For events $\boldsymbol{x}$ off the worldline one has
\begin{equation}\label{eq:n1}
x^a  \, = \, \xi ^a \big( \tau _R ( \boldsymbol{x} ) \big) \, + \,
r_R ( \boldsymbol{x} ) \Big( 
\dot{\xi}{}^a \big( \tau _R ( \boldsymbol{x} ) \big)
+ n^a ( \boldsymbol{x} ) \Big)
\end{equation}
with a well-defined spatial unit vector $\boldsymbol{n} (
\boldsymbol{x} )$, 
\begin{equation}\label{eq:n2}
n^a ( \boldsymbol{x} ) n_a ( \boldsymbol{x} ) \, = \, 1 \, ,
\qquad \quad
\dot{\xi}{}^a \big( \tau _R (\boldsymbol{x}) \big) 
n_a ( \boldsymbol{x} ) \, = \, 0 \, ,
\end{equation}
see Fig.~\ref{fig:retardation}. By differentiation, we find
\begin{equation}\label{eq:dtauR}
\partial _b \tau _R ( \boldsymbol{x} ) 
\, = \, - \, \dot{\xi}{}_b \big( \tau _R ( \boldsymbol{x} ) \big)
\, - \, n_b ( \boldsymbol{x} )
\end{equation}
and
\begin{equation}\label{eq:drR}
\partial _b r _R ( \boldsymbol{x} ) 
\, = \, n_b ( \boldsymbol{x} ) \, + \, 
r_R ( \boldsymbol{x} ) 
\ddot{\xi}{}^a \big( \tau _R ( \boldsymbol{x} ) \big)
n_a ( \boldsymbol{x} ) \Big(
\dot{\xi}{}_b \big( \tau _R ( \boldsymbol{x} ) \big)
+
n_b ( \boldsymbol{x} ) \Big) \, .
\end{equation}
For a sequence \cRed{of events} $\boldsymbol{x} _N$
that approaches the worldline, $\boldsymbol{x} _N \to 
\boldsymbol{ \xi} (\tau _0 )$ as $N \to \infty$, 
$\tau _R ( \boldsymbol{x} _N) \to \tau _0$ and $r_R ( \boldsymbol{x} _N) 
\to 0$. \cRed{The limits for}  
$\boldsymbol{n} ( \boldsymbol{x} _N)$,
$\partial _b \tau _R ( \boldsymbol{x} _N)$ 
and $\partial _b r _R ( \boldsymbol{x} _N)$ \cRed{do not exist as $N\to\infty$}.

\cRed{We model}
a point charge with worldline $\boldsymbol{\xi} (\tau)$
\cRed{by the distributional current density}
\begin{equation}\label{eq:jpoint}
j_a( \boldsymbol{x} )  \, = \, 
q \int _{\tau _{\mathrm{min}}} ^{\tau _{\mathrm{max}}}
\delta \big( \boldsymbol{x}- \boldsymbol{\xi} ( \tau) \big) \, 
\dot{\xi}{}_a ( \tau ) \, d \tau 
\end{equation}
where $q$ is \cRed{its electric} charge. Then (\ref{eq:Ahatret}) gives 
the standard Li{\'e}nard-Wiechert potential,
\begin{equation}\label{eq:LW}
\hat{A}{}_a  ( \boldsymbol{x} ) \, = \,  
\frac{
q \, \dot{\xi}{}_a \big( \tau _R ( \boldsymbol{x} ) \big)
}{
r_R( \boldsymbol{x} )
} \, , 
\end{equation}
and (\ref{eq:Aret}) reads
\begin{equation}\label{eq:Apoint1}
A_a  ( \boldsymbol{x} ) \, = \, 
\frac{q}{\ell} \, 
\int _{\tau _{\mathrm{min}}} ^{\tau _R ( \boldsymbol{x} )}
{\frac{
J_1 \Big(\ell ^{-1} 
D\big( \boldsymbol{x}- \boldsymbol{\xi} ( \tau ) \big) \Big)
}{
D \big( \boldsymbol{x} - \boldsymbol{\xi} ( \tau ) \big)
}}
\, \dot{\xi}{} _a ( \tau ) \, d \tau 
\; .
\end{equation}
Both (\ref{eq:LW}) and (\ref{eq:Apoint1}) \cRed{are defined} 
for \cRed{events $\boldsymbol{x}$ off the worldline}.
In Section \ref{sec:pointon} we investigate what happens
\cRed{to the potentials and fields as} $\boldsymbol{x}$ approaches the worldline. 
\begin{figure}[h]
    \psfrag{x}{{$\boldsymbol{x}$}} 
    \psfrag{p}{{$\,$ \hspace{-0.4cm} $\boldsymbol{\xi} \big( \tau _R ( \boldsymbol{x}) \big)$}} 
    \psfrag{q}{{$\,$ \hspace{-0.8cm} $r_R ( \boldsymbol{x} ) \, \dot{\boldsymbol{\xi}} \big( \tau _R (\boldsymbol{x}) \big)$}} 
    \psfrag{n}{{$r_R ( \boldsymbol{x} ) \, \boldsymbol{n} ( \boldsymbol{x} ) $}} 
\centerline{\includegraphics[width=6.5cm]{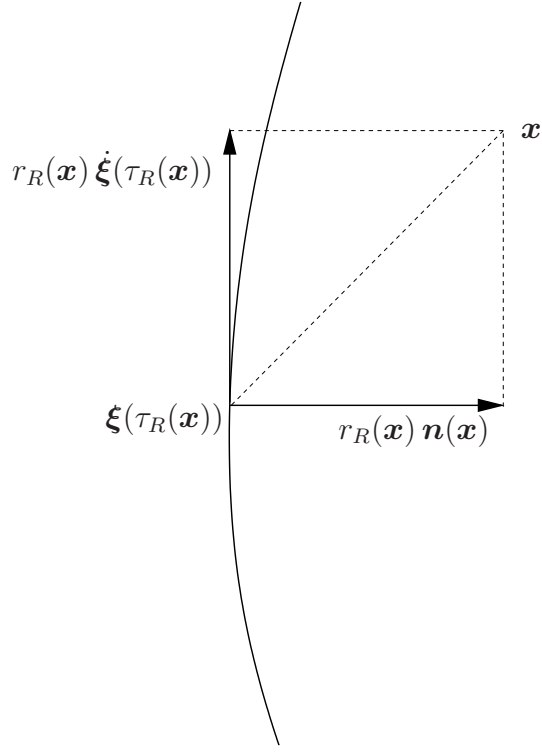}}
\caption{Retarded time and retarded distance}
\label{fig:retardation}
\end{figure}

Note that, in contrast to the Li{\'e}nard-Wiechert potential
(\ref{eq:LW}), the potential (\ref{eq:Apoint1}) at an event
$\boldsymbol{x}$ depends on the whole history of the charge from $\tau
= \tau _{\mathrm{min}}$ up to $\tau = \tau _R ( \boldsymbol{x}
)$. (The same is true, in general, \cRed{in} the standard
\cRed{Maxwell-Lorentz theory with point charges} on curved
spacetimes.) The integrand in (\ref{eq:Apoint1}) is bounded for $\tau
\to \tau _R ( \boldsymbol{x} )$ because
\begin{equation}\label{eq:tautotauR}
\frac{
J_1 \Big(\ell ^{-1} 
D\big( \boldsymbol{x}- \boldsymbol{\xi} ( \tau ) \big) \Big)
}{
D \big( \boldsymbol{x} - \boldsymbol{\xi} ( \tau ) \big)
}
\, \dot{\xi}{} _a ( \tau ) \, \longrightarrow \,
\frac{
\dot{\xi}{}_a \big( \tau _R ( \boldsymbol{x} ) \big)
}{
2 \ell
}
\qquad \textup{for} \quad 
\tau \to \tau _R ( \boldsymbol{x} )
\end{equation}
where we have used the Bernoulli-l'H{\^o}pital rule and 
$J_1'(0) = 1/2$. By contrast, for $\tau \to \tau _{\mathrm{min}}$,
the individual components $\dot{\xi}{}_a ( \tau )$ may blow
up arbitrarily. Therefore, the existence of the integral on the
right-hand side of (\ref{eq:Apoint1}) is not guaranteed. We show
later that, for a fairly large class of worldlines, this
integral does converge even absolutely, as a Lebesgue integral
or as an improper Riemann integral, for all $\boldsymbol{x}$ in 
the future of the worldline; however, we also 
give a (contrived) example where it does \emph{not} converge for 
some $\boldsymbol{x}$ in the future of the worldline.

On the assumption that the integral converges, at all events
$\boldsymbol{x}$ \cRed{we} can differentiate (\ref{eq:Apoint1})
\cRed{off the worldline} with respect to $x^b$.  Antisymmetrizing the
resulting expression gives the field strength (\ref{eq:FA}),
\begin{eqnarray}
\label{eq:Fpoint0}
\fl
F_{ab}  ( \boldsymbol{x} ) & = & 
\frac{q}{2 \ell ^2} \, 
\Big( \dot{\xi}{}_b \big( \tau _R ( \boldsymbol{x} ) \big)
n_a ( \boldsymbol{x} ) 
- 
\dot{\xi}{}_a \big ( \tau _R ( \boldsymbol{x} ) \big) 
n_b ( \boldsymbol{x} ) \Big) 
\\
\nonumber
\fl&&- \, \frac{q}{\ell ^2} 
\int _{\tau _{\mathrm{min}}} ^{\tau _R ( \boldsymbol{x} )}
\frac{
J_2 \Big(\ell ^{-1} 
D\big( \boldsymbol{x}- \boldsymbol{\xi} ( \tau ) \big) \Big)
}{
D \big( \boldsymbol{x} - \boldsymbol{\xi} ( \tau ) \big) ^2
}
\, \Big( \big( x_b -\xi _b ( \tau  ) \big) \dot{\xi}{}_a ( \tau)
-
\big( x_a -\xi _a ( \tau  ) \big) \dot{\xi}{}_b ( \tau) \Big)
\, d \tau
\; .
\end{eqnarray}
Here we have used (\ref{eq:dtauR}) and the identities 
$2 J_1 (z) = z \big( J_0 (z) + J_2 (z) \big)$ and 
$2 J_1' (z) =  J_0 (z) - J_2 (z)$ of the Bessel
functions. Again, we postpone the discussion of what happens 
if $\boldsymbol{x}$ approaches the worldline to Section \ref{sec:pointon}.

We observe that, keeping $\boldsymbol{x}$ fixed, we may use 
$\zeta= D \big( \boldsymbol{x} - \boldsymbol{\xi} (\tau) \big)$ 
as the parameter along the worldline. Indeed, differentiation of 
the equation
\begin{equation}\label{eq:zeta}
\zeta ^2 \, = \, - \, 
\big( x^a- \xi ^a ( \tau ) \big) \big( x_a- \xi _a ( \tau ) \big) 
\end{equation} 
yields
\begin{equation}\label{eq:zetatau}
\zeta \, d \zeta \, = \,  
\dot{\xi}{}^a ( \tau )  \big( x_a- \xi _a ( \tau ) \big) d \tau. 
\end{equation} 
As, by the \cRed{reverse} Schwartz inequality for timelike future-oriented
vectors,
\begin{equation}\label{eq:schwarz}
\dot{\xi}{}^a ( \tau ) \big( x_a- \xi _a ( \tau ) \big) 
< - \zeta < 0 \qquad  \textup{for} \quad \,  0 < \zeta < \infty \;,  
\end{equation} 
$\zeta$ is monotonically decreasing along the worldline. This
guarantees that the equation $\zeta = D
\big( \boldsymbol{x}- \boldsymbol{\xi} ( \tau ) \big)$ can
be solved for $\tau$,
\begin{equation}\label{eq:solvetau}
\zeta = D\big( \boldsymbol{x}- \boldsymbol{\xi} ( \tau ) \big)
\qquad \Longleftrightarrow \qquad 
\tau = \varpi ( \zeta , \boldsymbol{x} ) \; .
\end{equation} 
\cRed{As} proper time $\tau$ runs from \cRed{$\tau_{\min}$} to $\tau
_R (x)$, the new parameter $\zeta$ runs (backwards) from $\infty$ to
0. Hence \cRed{at all events $\boldsymbol{x}$} (\ref{eq:Apoint1}) can
be rewritten as
\begin{equation}\label{eq:Apoint2}
A_a  (\boldsymbol{x}) \, = \, - \, \frac{q}{\ell} \, 
\int _0 ^{\infty}
\left.
\frac{J_1 \big( \zeta / \ell  \big)
\, \dot{\xi}{} _a ( \tau )
}{
\dot{\xi}{}^b (  \tau ) 
\big( x_b- \xi _b (  \tau ) \big)} 
\right| _{\tau = \varpi ( \zeta, \boldsymbol{x} )}
 \, d \zeta 
\; .
\end{equation}
Note that, if we \cRed{regard} $\boldsymbol{x}$ as a parameter, (\ref{eq:Apoint2}) has the
form of a Hankel transform which transforms a function of $\zeta$ to a 
function of $1/\ell$. 

We may use the parameter $\zeta$ in the formula for the field strength
as well. If such a change of the integration variable is performed on
the right-hand side of (\ref{eq:Fpoint0}), the resulting equation reads
\begin{eqnarray}\label{eq:Fpoint1}
\fl F_{ab}  ( \boldsymbol{x} ) & = &  
\frac{q}{2 \ell ^2} \, 
\Big( \dot{\xi}{}_b \big( \tau _R ( \boldsymbol{x} ) \big)
n_a ( \boldsymbol{x} ) 
- 
\dot{\xi}{}_a \big ( \tau _R ( \boldsymbol{x} ) \big) 
n_b ( \boldsymbol{x} ) \Big) 
\\[0.2cm]
\nonumber
\fl &&+ \, \frac{q}{\ell ^2} 
\int _{0} ^{\infty}
\left.
\frac{
\Big( \big( x_b -\xi _b ( \tau  ) \big) \dot{\xi}{}_a ( \tau)
-
\big( x_a -\xi _a ( \tau  ) \big) \dot{\xi}{}_b ( \tau) \Big)
}{
\dot{\xi}{}^c ( \tau) 
\big( x_c -\xi _c ( \tau  ) \big)
} 
\right|
_{\tau = \varpi ( \zeta , \boldsymbol{x} )}
\frac{
J_2 (\zeta / \ell ) \, d \zeta
}{
\zeta
}
\; .
\end{eqnarray}
\cRed{An} alternative expression for the field strength \cRed{is
  obtained if, off
  the worldline,}
we differentiate (\ref{eq:Apoint2}) with respect to $x^b$
and antisymmetrize,
\begin{eqnarray}\label{eq:Fpoint2}
\fl F_{ab}  (\boldsymbol{x}) & = & 
\frac{q}{\ell} \, 
\int _0 ^{\infty}
\left.
\frac{
J_1 \big( \zeta / \ell  \big)
\Big( \ddot{\xi}{} _a ( \tau )
\big( x_b - \xi _b ( \tau ) \big)
-
\ddot{\xi}{} _b ( \tau )
\big( x_a - \xi _a ( \tau ) \big)
\Big) 
}{
\Big( \dot{\xi}{}^c (  \tau ) 
\big( x_c- \xi _c (  \tau ) \big) \Big) ^2
}
\right| _{\tau = \varpi ( \zeta, \boldsymbol{x} )}
 \! \!  d \zeta 
\\
\nonumber
\fl && \hspace{-4em}- \, \frac{q}{\ell} \! \! 
\int_0 ^{\infty}
\left.
\frac{
J_1 \big( \zeta / \ell  \big)
\!
\Big( 
\!
\dot{\xi}{} _a ( \tau )
\big( x_b - \xi _b ( \tau ) \big)
\! - \!
\dot{\xi}{} _b ( \tau )
\big( x_a - \xi _a ( \tau ) \big)
\!
\Big)
\Big( 
\!
1 + \ddot{\xi}{}^d ( \tau ) 
\big( x_b- \xi _b (  \tau ) \big) 
\!
\Big)
}{
\Big( \dot{\xi}{}^c (  \tau ) 
\big( x_c - \xi _c (  \tau ) \big) \Big) ^3
}
\right| _{\tau = \varpi ( \zeta, \boldsymbol{x} )}
\! \!  \! \! \! \! d \zeta 
\; .
\end{eqnarray}
\cRed{Alternatively (\ref{eq:Fpoint2}) can be derived directly from
  (\ref{eq:Fpoint1}) by integrating its second term by parts.}

\cRed{It is worth noting} from (\ref{eq:Apoint2}) that we can
 derive another form of the 
potential by performing an integration by parts and using 
the identity $-J_1=J_0'$ of Bessel functions. The resulting equation
\begin{eqnarray}\label{eq:Apoint3}
\fl A_a ( \boldsymbol{x}) & = & \frac{q \, 
\dot{\xi}{}_a ( \tau _R ( \boldsymbol{x} ) \big)}{r_R ( \boldsymbol{x} )} 
\\ \fl &&
 - \, 
q \, \int _0 ^{\infty} \
\frac{ 
\ddot{\xi}{}_a (  \tau )
\, - \,
\frac{\displaystyle 
\dot{\xi}{}_a (  \tau ) 
\Big( 1 + \ddot{\xi}{}^d (  \tau )  
\big( x_d - \xi _d (  \tau ) \big) \Big)
}{\displaystyle
 \dot{\xi}{}^c (  \tau ) 
\big( x_c - \xi _c (  \tau ) \big) 
} 
}{
\Big( \dot{\xi}{}^b (  \tau ) 
\big( x_b - \xi _b (  \tau ) \big) \Big) ^2
}
\left.
\rule{0em}{2em}
\right| _{\tau = \varpi ( \zeta, \boldsymbol{x} )}
\, J_0 \big( \zeta/ \ell \big) \, \zeta  \, d \zeta 
\nonumber
\end{eqnarray}
gives the deviation of the potential from the Li{\'e}nard-Wiechert potential.
With the help of standard asymptotic formulas for the Bessel function $J_0$ 
the right-hand side can be written as a power series in $\ell$.
Such asymptotic (i.e., in general non-convergent) expansions have 
been used, e.g., by \citeasnoun{Frenkel1996} (also 
see \citeasnoun{FrenkelSantos1999}) and \citeasnoun{Zayats2013}. 

\subsection*{Example 1: Charge at rest}

The simplest case one can consider is a charge \cRed{at} 
rest in an appropriately chosen inertial system; this is 
equivalent to saying that the worldline of the charge is 
a straight timelike line, 
\begin{equation}\label{eq:straight}
\xi ^a ( \tau) = V^a \, \tau 
\end{equation}
with a constant four-vector $\boldsymbol{V}$ 
satisfying $V_aV^a=-1$. In this case (\ref{eq:zeta})
is a quadratic equation,
\begin{equation}\label{eq:zetastat1}
\zeta ^2
\, = \, \tau ^2 - 2 \, \tau \big( r_R ( \boldsymbol{x} )  - 
\tau _R ( \boldsymbol{x} ) \big) \, + \, 2 \,  
\tau _R ( \boldsymbol{x} ) r_R ( \boldsymbol{x} )
+ \tau _R ( \boldsymbol{x} ) ^2 \, ,
\end{equation}
which can be easily solved for $\tau$,
\begin{equation}\label{eq:zetastat}
\varpi ( \zeta, \boldsymbol{x} ) 
\, = \, r_R ( \boldsymbol{x} ) \, + \, 
\tau _R ( \boldsymbol{x} ) \, - \, 
\sqrt{\zeta ^2 + r_R ( \boldsymbol{x} )^2}
\, .
\end{equation}
Then (\ref{eq:Apoint2}) simplifies to
\begin{equation}\label{eq:Astat1}
A_a  (\boldsymbol{x}) \, = \, - \, \frac{q \, V^a}{\ell}  \,
\int _0 ^{\infty}
\frac{J_1 \big( \zeta / \ell  \big)
\,  d \zeta
}{
\sqrt{\zeta ^2 + r_R ( \boldsymbol{x} )^2}
}
\,  = \, 
\frac{q \, V_a}{r_R (\boldsymbol{x})} \,
\Big(1-e^{-r_R (\boldsymbol{x})/\ell} \Big)
\end{equation}
which is finite for all $\boldsymbol{x}$,
\begin{equation}\label{eq:Astat2}
A_a ( \boldsymbol{x}) \,  \rightarrow \, - \,  
\frac{q \, V_a }{\ell} \qquad \textup{for} \quad 
r_R (\boldsymbol{x} ) \to 0 \, .
\end{equation}
By differentiation of (\ref{eq:Astat1}) \cRed{off the worldline}, or
equivalently by evaluation of (\ref{eq:Fpoint1}) \cRed{or
  (\ref{eq:Fpoint2}) off the worldline}, we get the field
strength
\begin{eqnarray}\label{eq:Fstat}
\fl
F_{ab} ( \boldsymbol{x}) \,  &=& \, 
\frac{q}{r_R( \boldsymbol{x} )^2} 
\,
\Big( n_a ( \boldsymbol{x} )V_b-n_b( \boldsymbol{x} )V_a \Big)
\left( \, 1 \, - \, e^{-r_R ( \boldsymbol{x} ) / \ell}
- \frac{r_R(\boldsymbol{x})}{\ell} \, e^{-r_R ( \boldsymbol{x})/\ell}
\, \right) 
\nonumber
\\
\fl&=& \, 
\frac{q}{2 \ell ^2} 
\,
\Big( n_a ( \boldsymbol{x} )V_b-n_b( \boldsymbol{x} )V_a \Big)
\Big( \, 1 \, + \, O \big( r_R ( \boldsymbol{x} )\big) \Big)
\, .
\end{eqnarray}
This expression is not defined on the worldline; if a point on the
worldline is approached, the limit \cRed{of some components} depend on
the direction \cRed{and we say} that the field displays a
\emph{directional singularity}.  \cRed{Although for events
  $\boldsymbol{x}$ on the worldline the field (\ref{eq:Fpoint1}) is
  undefined, the field (\ref{eq:Fpoint2}) can be evaluated there and
  yields $F_{ab} (\boldsymbol{x} ) = 0$ for all $a,b$. This value also
  arises by a certain averaging procedure described in the next
  section.}

In the rest system of the charge we have $V_{\mu}=0$ for 
$\mu = 1,2,3$ and $r_R ( \boldsymbol{x} ) = 
\sqrt{(x^1)^2 +(x^2)^2+(x^3)^2}$ is just the ordinary radius 
coordinate. In this coordinate system (\ref{eq:Fstat}) gives
a radial electrostatic field with \cRed{{\em modulus}}
\begin{equation}\label{eq:Estat1}
E(r) \, = \, \frac{q}{r^2} \,
\left( \, 1 \, - \, e^{-r / \ell}
- \frac{r}{\ell} \, e^{-r/\ell}
\, \right) \, = \, \frac{q}{2 \ell ^2} 
\Big( 1+O(r) \Big) \, .
\end{equation}
In contrast to the Coulomb field of the standard 
Maxwell theory, the Bopp-Podolsky $E(r)$ stays finite 
as the worldline is approached. 
This result played a crucial role in the original work 
of \citeasnoun{Bopp1940} and \citeasnoun{Podolsky1942}.
It has the consequence that, at least for a charge
at rest, the  total field energy\footnote{The
    total field energy 
${\cal{E}}=\int_{\mathbf{R}^3} T_{00} d^3x$ where $T_{ab}$ is defined in the
appendix.} is finite. Note, however, that the electric 
\cRed{{\em vector field}} cannot be continuously extended into the
origin, because of the above-mentioned directional
singularity, see Fig.~\ref{fig:static2}. 


\begin{figure}[h]
    \psfrag{x}{{$x$}} 
    \psfrag{y}{{$E_x (x)$}} 
    \psfrag{a}{{$\,$ \hspace{-0.75cm} $\frac{q}{2 \ell ^2}$}} 
    \psfrag{b}{{$\,$ \hspace{0.3cm} $-\frac{q}{2 \ell ^2}$}} 
\centerline{\includegraphics[width=13.5cm]{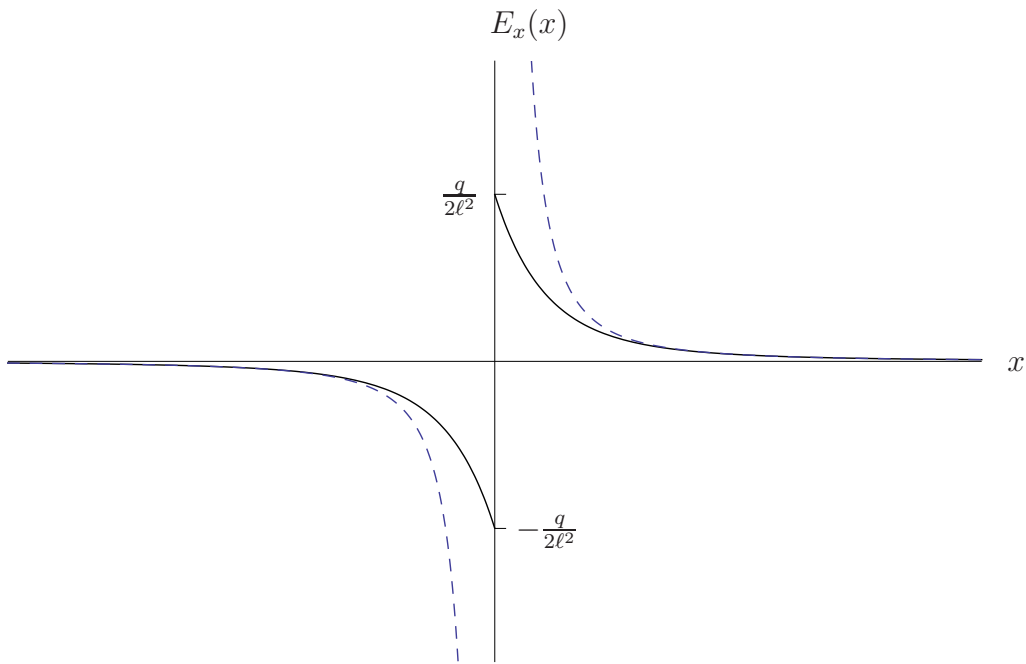}}
\caption{The $x$-component of the electric field strength on the $x$-axis
of a charge at rest in the Bopp-Podolsky theory (solid) and in the 
standard Maxwell theory (dashed)}
\label{fig:static2}
\end{figure}

\section{Field of a point charge on the worldline and 
self-force}\label{sec:pointon}

The two expressions (\ref{eq:Apoint1}) and (\ref{eq:Apoint2}) for the
potential are equivalent \cRed{off} the worldline, and so are
the two expressions (\ref{eq:Fpoint1}) and (\ref{eq:Fpoint2}) for the
field strength.  We now discuss \cRed{their behavior on and near
the worldline}. This is crucial because the value of the field strength
on the worldline \cRed{will determine} the self-force. 
\cRed{Both representations of the potential (\ref{eq:Apoint1}) and
  (\ref{eq:Apoint2}) are well defined and continuous on the
  worldline.} If \cRed{the potential}
were differentiable, its derivative would give the field strength
on the worldline without any ambiguity. However, \cRed{it} is
\emph{not} differentiable on the worldline. This is the reason why the
expression (\ref{eq:Fpoint1}), which results from differentiating
(\ref{eq:Apoint1}), and the expression (\ref{eq:Fpoint2}), which
results from differentiating (\ref{eq:Apoint2}), behave differently on
the worldline.  

\cRed{First} observe that (\ref{eq:Fpoint1}) is not defined at
events on the worldline because it involves the derivative $\partial
_b \tau _R ( \boldsymbol{x} )$ which is not defined at such events. By
contrast, (\ref{eq:Fpoint2}) does not involve this derivative and
\cRed{gives a {\em unique}} value for the field strength on the worldline,
provided that the integrals converge. However, this value does not
result from differentiating the potential on the worldline because in
order to differentiate (\ref{eq:Apoint2}) \cRed{one must pass} the
derivative \cRed{under} the integral; this is only valid if the integrand
has a continuous derivative, which does not occur if $\boldsymbol{x}$
is on the worldline.

\cRed{ As noted, the field $F_{ab}(\boldsymbol{x})$ given by (\ref{eq:Fpoint1})
  or (\ref{eq:Fpoint2}) is not continuous in a neighborhood of events
  on the worldline: if a point $\boldsymbol{x}$ on the worldline is
  approached some of the components have non-zero {\em finite}
  limits that depend on
  the direction taken. This is manifest in (\ref{eq:Fpoint1}) where
  the discontinuity arises solely from the first term. Such behavior
  contrasts with similar problems in defining the self-force on a
  charged particle in Maxwell-Lorentz electrodynamics where a similar
  directional dependence arises as one approaches the particle
  worldline. However, in that theory the same limiting values are
  infinite.}

\cRed{Since the self-force (to be defined below) depends on the values
  of $F_{ab}(\boldsymbol{x})$ on the worldline, one must assign them
  {\em specific} values. Such assignments should be based on physical
  criteria beyond the mathematical analysis thus far. For example an
  isolated free point charge should remain in ``inertial'' motion in
  the absence of external fields, i.e. the self-force should be
  zero. From example 1, this can be achieved by assigning the value
  zero to all $F_{ab}(\boldsymbol{x})$ on the inertial worldline in
  this case. This assignment is equivalent to either adopting
  (\ref{eq:Fpoint2}) directly for $F_{ab}(\boldsymbol{x})$ or applying
  a ``directional-averaging'' (see below) to
  (\ref{eq:Fpoint1}). Although for arbitrary motion one may always
  define a directional-averaging such that the value given by
  (\ref{eq:Fpoint2}) arises from such a process applied to
  (\ref{eq:Fpoint1}), this is by no means a unique procedure.  In
  general one may consider more complex averaging procedures involving
  constructions based on extensions of the natural Frenet frame
  defined by the worldline.}

\begin{figure}
\psfrag{xit}{$\boldsymbol{\xi}(\tau)$}
\psfrag{xitdot}{$\dot{\boldsymbol{\xi}}(\tau_0)$}
\psfrag{xit0}{$\boldsymbol{\xi}(\tau_0)$}
\psfrag{xit1}{$\boldsymbol{\xi}(\tau_1)$}
\psfrag{xit2}{$\boldsymbol{\xi}(\tau_2)$}
\centerline{\includegraphics[height=12.cm]{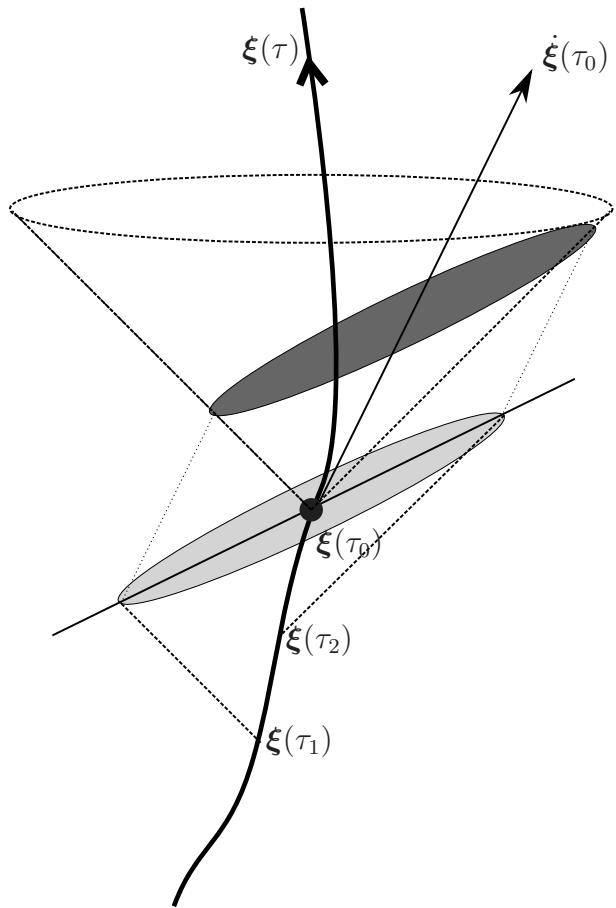}}
\caption{\cRed{This figure illustrates a retarded 2-sphere (dark grey)
    and a rest 2-sphere (light grey), defined by an arbitrary timelike
    worldline in Minkowski spacetime. The rest sphere with origin
    $\boldsymbol{\xi}(\tau_0)$ lies in the orthocomplement of the
    tangent vector $\dot{\boldsymbol{\xi}}(\tau_0)$. The retarded
    sphere is the intersection of the rest sphere's parallel
    translation along $\dot{\boldsymbol{\xi}}(\tau_0)$ with the
    forward lightcone of ${\boldsymbol{\xi}}(\tau_0)$. The averaging
    procedure described in the text corresponds to sampling the fields
    of the particle as the radii of these spheres tend to zero. Fields
    on the rest sphere are generated by events on the past history of
    $\boldsymbol{\xi}(\tau_2)$. By contrast fields on the retarded
    sphere are all generated by the past history of
    $\boldsymbol{\xi}(\tau_0)$.}}
\label{fig:Spheres}
\end{figure}

{
A procedure based solely on the
worldline's tangent vector, rather {than} the full geometry of its Frenet
frame can 
be done most easily if we introduce the (spatial)
{\em retarded 2-sphere}, see Fig. \ref{fig:Spheres},
\begin{equation}\label{eq:retsph}
S(\tau _0 , r_0) = \big\{ \boldsymbol{x} \in \mathbb{R}{}^4
\big| \, \tau _R ( \boldsymbol{x} ) = \tau _0 \, , \:
r_R ( \boldsymbol{x} ) = r_0 \big\} \, .
\end{equation}
The directional-average of any spacetime tensor's inertial components
$K^{a\ldots}{}_{b\ldots}(\boldsymbol{x})$ over the 2-sphere
$S(\tau _0, r_0 )$ is defined as
\begin{equation}
\overline{K^{a\ldots}{}_{b\ldots}}=\frac{1}{4\pi}\int\!\!\int_{S(\tau _0, r_0 )} 
K^{a\ldots}{}_{b\ldots}(\boldsymbol{x}) \, d S
\label{eq:average} 
\end{equation}
where $d S$ is the natural surface measure on $S(\tau _0, r_0 )$,
induced by the ambient Minkowski metric.
This integral depends only on $\tau_0$, $r_0$,
$\boldsymbol{\xi}(\tau_0)$ and $\dot{\boldsymbol{\xi}}(\tau_0)$ and 
is manifestly Lorentz {\em covariant} with respect to global
Lorentz transformations. I.e. if
\begin{equation}
K^{a'\ldots}{}_{b'\ldots}(\boldsymbol{x}) = 
\Lambda^{a'}{}_{a} \cdots\,\Lambda_{b'}{}^{b} \cdots\,
K^{a\ldots}{}_{b\ldots}
(\boldsymbol{x})
\end{equation}
then
\begin{equation}
\overline{K^{a'\ldots}{}_{b'\ldots}}
=
\Lambda^{a'}{}_{a} \cdots\Lambda_{b'}{}^{b} \cdots 
\overline{K^{a\ldots}{}_{b\ldots}}
\end{equation}
where $\Lambda^{a'}{}_{a}\in SO(3,1)$.

Coordinating $S(\tau _0, r_0 )$ with spherical polars
$(\vartheta,\varphi)$, the measure $dS=\sin\vartheta\,d\vartheta\,d\varphi$ 
and 
(\ref{eq:average}) 
becomes
\begin{equation}
\overline{K^{a\ldots}{}_{b\ldots}}=\frac{1}{4\pi}
\int _0 ^{2\pi}
\int _{0} ^{\pi} 
K^{a\ldots}{}_{b\ldots} (\boldsymbol{x})\,
\mathrm{sin} \, \vartheta \, d \vartheta \, d \varphi 
\label{eq:average2} 
\end{equation}

In terms  of an orthonormal 
tetrad $\big( \boldsymbol{e}{}_0 , \boldsymbol{e}{}_1 ,
\boldsymbol{e}{}_2 , \boldsymbol{e}{}_3 \big)$ with
$\boldsymbol{e}{}_0 = \dot{\boldsymbol{\xi}} ( \tau _0 )$
\begin{equation}\label{eq:ne}
n^a ( \boldsymbol{x} ) = 
\mathrm{cos} \, \varphi  \, \mathrm{sin} \, \vartheta \, e_1^a   
+
\mathrm{sin} \, \varphi  \, \mathrm{sin} \, \vartheta \, e_2^a   
+
\mathrm{cos} \, \vartheta \, e_3^a
\end{equation}
for $\boldsymbol{x}\in S(\tau _0, r_0 )$. Hence
$\overline{n^a}=0$. Since $\dot{\xi}{}^a \big( \tau _R (
\boldsymbol{x} ) \big) = \dot{\xi}{}^a ( \tau _0 \big)$ is constant on
$S( \tau _0 , r_0 )$ the first term on the right-hand side of
(\ref{eq:Fpoint1}) averages to zero. Hence its limit for $r_0 \to 0$
is zero as well, so this term gives no contribution to the average of
the field (\ref{eq:Fpoint1}) at $\boldsymbol{\xi} ( \tau _0 )$. Since
the second term in (\ref{eq:Fpoint1}) is continuous in a neighborhood
of the $\boldsymbol{\xi}(\tau_0)$ its average is well defined. Hence
the averaged value of (\ref{eq:Fpoint1}) equals the second term on the
right hand side of (\ref{eq:Fpoint1}) and is also equal to
(\ref{eq:Fpoint2}) after an integration by parts. 

An alternative averaging procedure can be defined in terms of 
{\em rest} spheres in the orthocomplement of $\dot{\xi} (\tau _0 )$, see Fig. \ref{fig:Spheres}. In this
case $\dot{\xi}{}^a \big( \tau_R ( \boldsymbol{x} ) \big)$ is
\emph{not} constant on a sphere of finite radius $r_0$, so the 
calculation is less convenient; however, in the limit of $r_0$
tending to zero one finds, again, that the first term on
the right-hand side of (\ref{eq:Fpoint1}) averages to zero.
Thus these two averaging procedures give the same result and 
correspond to surrounding the 
worldline either by a \emph{Bhabha tube} or by a \emph{Dirac
tube}. See \citeasnoun{Norton2009} and \citeasnoun{FerrisGratus2011} for
a similar discussion in the context of the standard Maxwell-Lorentz
  theory with point charges.

Motivated by the application of this procedure to inertial motion in
example 1 and the equivalence of the directional-average of
(\ref{eq:Fpoint1}) with (\ref{eq:Fpoint2}) for general motion, we 
assign to $F_{ab}(\boldsymbol{x})$ on the
worldline the unique value (\ref{eq:Fpoint2}) for {\em all
  motions}. This is intuitively persuasive if one thinks of the point
charge as the limiting case of an extended charge whose size tends to
zero. Directional-averaging is formulated as an axiom in the living
review on the self-force by Poisson, Pound and Vega, see Section 24.1
in \cite{PoissonPoundVega2011}.}

\cRed{Given a definition of a} {\em finite} field $F_{ab} \big( \boldsymbol{x} \big)$
at an event $\boldsymbol{x} = \boldsymbol{\xi} ( \tau _0 )$
on the worldline, \cRed{the self-force $f^{\textup{s}}_a ( \tau _0 )$ is defined
as}
the \cRed{relativistic} Lorentz force\footnote{In the appendix, this force is shown to
  arise from the divergence of the electromagnetic
  stress-energy-momentum tensor associated with the Bopp-Podolsky theory.}
exerted by \cRed{$F_{ab} \big( \boldsymbol{\xi} ( \tau _0 ) \big)$ on}
the point charge that produces the field: 
\begin{equation}\label{eq:sf1}
f^{\textup{s}}_a ( \tau _0 ) \,  = \, 
q \, F_{ab}  \big( \boldsymbol{\xi} ( \tau _0 ) \big) 
\, \dot{\xi}{}^b ( \tau _0 ) \, .
\end{equation}
The \cRed{equation} of motion for some $C^0$ functions $\xi^a(\tau)$
is \cRed{assumed to take the form}
\begin{equation}\label{eq:eom}
m \, \ddot{\xi}{} _a ( \tau ) \,  = \, 
f_a^{\textup{s}} ( \tau ) \, + \, f_a^{\textup{e}}  ( \tau )  
\end{equation}
where $m$ denotes the {\em finite} inertial mass of the particle 
 and $f_a^{\textup{e}}  ( \tau )$ is an external \cRed{Minkowski} force. If the 
latter is electromagnetic in origin,
$f_a^{\textup{e}} ( \tau )=q \, F^{\textup{e}}_{ab} \big( \boldsymbol{\xi} ( \tau ) \big)
\, \dot{\xi}{}^b ( \tau )$, were $F^{\textup{e}}_{ab}$ solves the
Bopp-Podolsky field equations with all sources \cRed{other}
than the point particle with charge $q$\cRed{, which includes the
  special case of no other sources.}
If $f_a^{\textup{e}} ( \tau )$ is given,
(\ref{eq:eom}) is an integro-differential equation for the worldline
$\boldsymbol{\xi} ( \tau )$.

In the standard \cRed{Maxwell-Lorentz theory with point charges}, the self-force is infinite; 
therefore, it is necessary to perform a mass renormalization, 
introducing a ``bare mass'' of the particle that is negative 
infinite.  By contrast, in the  
Bopp-Podolsky theory there is no need or justification
for introducing an infinite bare mass.

\cRed{From (\ref{eq:Fpoint2}) the self-force reads
\begin{equation}\label{eq:sf2a}
f^{\textup{s}}_a ( \tau _0 ) \,  = \, 
\frac{q^2 \dot{\xi}{}^b ( \tau _0 )}{\ell } 
\int _{0} ^{\infty}
\frac{\partial}{\partial \zeta}
W_{ab} \big( \zeta, \boldsymbol{\xi} ( \tau _0 ) \big) 
\frac{
J_1 (\zeta / \ell ) \, d \zeta
}{
\zeta
}
\, .
\end{equation}
where
\begin{equation}\label{eq:Wab}
\fl
W_{ab} \big( \zeta, \boldsymbol{\xi} ( \tau _0 ) \big) 
\, = \, 
\left.
\frac{
\big( \xi _b ( \tau _0 ) -\xi _b ( \tau  ) \big) \dot{\xi}{}_a ( \tau)
-
\big( \xi _a ( \tau _0 )-\xi _a ( \tau  ) \big) \dot{\xi}{}_b ( \tau) 
}{
\dot{\xi}{}^c ( \tau) 
\big( \xi _c ( \tau _0 ) -\xi _c ( \tau  ) \big)
} 
\right|
_{\tau = \varpi ( \zeta , \boldsymbol{\xi} ( \tau _0 ) )}
\, .
\end{equation}
}

\cRed{Although (\ref{eq:eom}) contains explicit derivatives of maximal
  order two, the presence of the integral (\ref{eq:sf2a}) over the
  past history of the worldline implies that it cannot be solved given
  only $f^{\textup{e}}_a(\tau)$ and the initial position and velocity
  of the particle at any initial $\tau$. Integro-differential
  equations involving retarded (or memory) effects are not uncommon in
  continuum mechanics and the Maxwell electrodynamics of continuous
  media. Indeed even in vacuo the modifications of the
  Abraham-Lorentz-Dirac equation due to the presence of a background
  gravitational field yields a similar integral over the past history
  of a charged point particle worldline \cite{DeWittBrehm1960}. In
  such problems additional physical criteria motivate analytic or
  numerical procedures that can be used to construct solutions. Since
  in principle the past history of any point charge is not empirically
  accessable, it seems inevitable that similar methodologies will be
  required to determine any future motion uniquely from (\ref{eq:eom})
  given only consistent field and particle data on some arbitrary
  3-dimensional spacelike hypersurface in Minkowski spacetime. These
  general ideas will be developed elsewhere.}

\subsection*{Example 2: Uniformly accelerated motion}

For a particle in \cRed{hyperbolic motion} with constant acceleration
$a$ \cRed{in its instantaneous rest frame}, the worldline is given by
\begin{eqnarray}\label{eq:hyp1}
\xi ^0 ( \tau ) 
\, = \, 
\frac{1}{a} \, \mathrm{sinh} (a \tau ) 
\, , \qquad
\xi ^1 ( \tau ) 
\, = \, 
\frac{1}{a} \, \mathrm{cosh} (a \tau ) \,
\\
\nonumber
\xi ^2 ( \tau ) 
\, = \, 
\xi ^3 ( \tau ) 
\, = \, 0 \, .
\end{eqnarray}
In this case, for a point $\boldsymbol{x} = \boldsymbol{\xi} 
( \tau _0 )$ on the worldline (\ref{eq:zeta}) reads
\begin{equation}\label{eq:hypzeta}
\zeta ^2 \, = \, \frac{2}{a^2} \, 
\Big( \mathrm{cosh} \big( a ( \tau _0 - \tau ) \big) 
\, - \, 1 \, \Big) \, .
\end{equation}
The self-force (\ref{eq:sf2a}) reduces to 
\begin{equation}\label{eq:sfh1}
f_a^{\textup{s}} ( \tau _0 ) \,  = \,  - \, q^2 \, \ddot{\xi}{}_a ( \tau _0 ) \, 
\int _0 ^{\infty}  
\frac{ J_2 \big( \frac{\zeta}{\ell} \big) \, d \zeta }{ 
2 \, \ell ^2  \sqrt{1+\frac{a^2 \zeta ^2}{4}}}
\end{equation}
which can be expressed in terms of the Bessel functions 
$I_1$ and $K_1$,
\begin{equation}\label{eq:sfh2}
f_a^{\textup{s}} ( \tau _0 ) 
\,  = \, 
 - \, \frac{q^2}{a \ell ^2} \, I_1 \big( (a \ell)^{-1} \big) \,
K_1 \big( (a \ell)^{-1} \big) \, \ddot{\xi}{}_a ( \tau _0 ) \, . 
\end{equation}
This result was recently found by \citeasnoun{Zayats2013}.
So the self-force \cRed{here is manifestly} finite.

\section{Finiteness of the field of a point charge and of the 
self-force}\label{sec:finite}
In the standard Maxwell-Lorentz theory, the (Li{\'e}nard-Wiechert) potential 
of a point charge and the corresponding field strength are singular on
the worldline of the source. By contrast, in
the Bopp-Podolsky theory there is a large class of worldlines for 
which the self-force is given by an absolutely converging
integral.

As $\dot{\boldsymbol{\xi}} ( \tau )$ \cRed{is a timelike unit vector}, and $\boldsymbol{x}- 
\boldsymbol{\xi}  ( \tau)$ has Lorentz length $\zeta$, we may write
\begin{equation}\label{eq:chi}
\dot{\xi}{}^a \big( \varpi ( \zeta , \boldsymbol{x} ) \big) 
\, = \,  
\mathrm{cosh} \, \chi ( \zeta , \boldsymbol{x} )  \, \delta _0^a
\, + \, 
\mathrm{sinh} \chi ( \zeta , \boldsymbol{x} )  
\, \nu ^{\rho} ( \zeta , \boldsymbol{x} ) \, \delta _{\rho} ^a
\end{equation}
and
\begin{equation}\label{eq:psi}
x^a - \xi ^a \big( \varpi ( \zeta , \boldsymbol{x} ) \big) 
\, = \,  
\zeta \, 
\Big( \mathrm{cosh} \, \psi ( \zeta , \boldsymbol{x} ) \, \delta _0^a
\, + \, 
\mathrm{sinh} \, \psi ( \zeta , \boldsymbol{x} ) 
\, \mu ^{\rho} (\zeta , \boldsymbol{x} ) \, \delta _{\rho} ^a \Big)
\end{equation}
where $\boldsymbol{\nu} ( \zeta , \boldsymbol{x} )$ and 
$\boldsymbol{\mu} ( \zeta , \boldsymbol{x} )$ are spatial unit vectors, 
\begin{equation}\label{eq:numu}
\nu _{\rho} (\zeta , \boldsymbol{x} )
\, \nu ^{\rho} (\zeta , \boldsymbol{x} ) 
\, = \, 
\mu _{\rho} (\zeta , \boldsymbol{x} )
\, \mu ^{\rho} (\zeta , \boldsymbol{x} )
\, = \, 1 \, . 
\end{equation}
We introduce the following terminology.
\begin{definition}\label{def:baftlc}
The worldline $\boldsymbol{\xi}$ is \emph{bounded away 
from the past light-cone} of an event $\boldsymbol{x}$ in the
future of $\boldsymbol{\xi}$ if 
$\psi \big( \zeta , \boldsymbol{x} )$ stays bounded
for $\zeta \to \infty$. 
\end{definition}
Geometrically, $\boldsymbol{\xi}$ is not bounded away 
from the past light-cone of $\boldsymbol{x}$ if and only if there 
is a sequence $\tau _k$ such that $\big(\boldsymbol{x} - \boldsymbol{\xi} 
(\tau _k) \big) / \big( x^0 - \xi^0 (\tau _k ) \big)$ approaches a 
lightlike vector for $\tau _k \to \tau _{\mathrm{min}}$, see 
Fig.~\ref{fig:baftlc}. 

\begin{figure}
    \psfrag{x}{{$\boldsymbol{x}$}} 
    \psfrag{p}{{$\boldsymbol{\xi} ( \tau )$}} 
\centerline{\includegraphics[width=15cm]{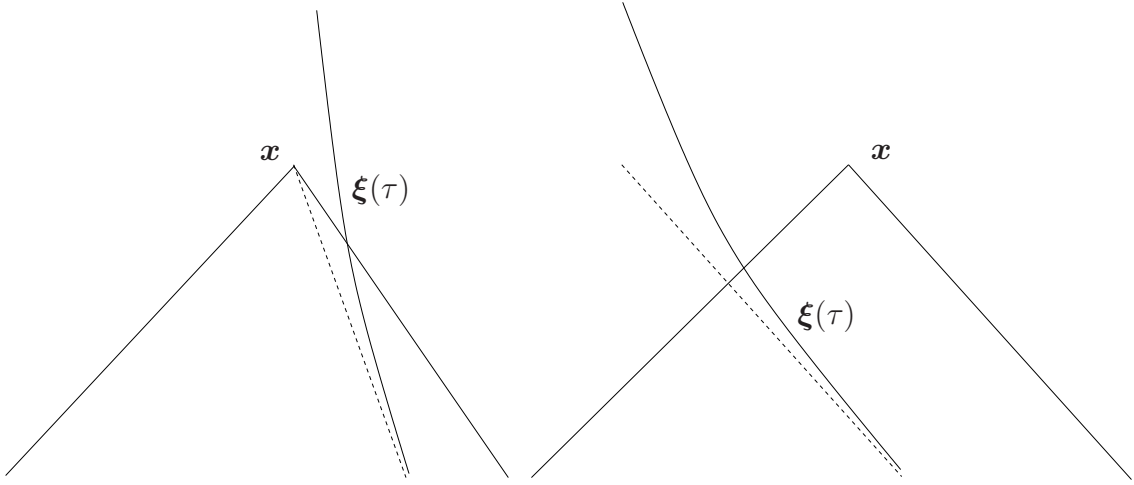}}
\caption{Worldline bounded away from the past light-cone of 
$\boldsymbol{x}$ (left) and not bounded away from the past 
light-cone of $\boldsymbol{x}$ (right)}
\label{fig:baftlc}
\end{figure}

The notion of being bounded away from the light-cone implicitly
refers to a particular inertial coordinate system chosen, and it
refers to a particular event $\boldsymbol{x}$. However, the notion
is actually independent of these choices, as the following proposition
shows.
\begin{proposition}\label{prop:baftlc}
The property of the worldline being bounded away from the past 
light-cone of $\boldsymbol{x}$ is preserved if we change the inertial
coordinate system by an orthochronous Lorentz transformation.
If this property is true for one event $\boldsymbol{x}$ in the 
future of the worldline, then the
future of the worldline is all of $\mathbb{R}^4$ and the 
property is true for all other events $\boldsymbol{y} \in 
\mathbb{R}^4$ as well.  
\end{proposition}
\begin{proof}
{}From (\ref{eq:psi}) we read that
\begin{equation}\label{eq:ba1}
\frac{x^0 - \xi ^0 \big( \varpi ( \zeta , \boldsymbol{x} ) \big) 
}{\sqrt{ - \Big(
x^a - \xi ^a \big( \varpi ( \zeta , \boldsymbol{x} ) \big) 
\Big) \Big(
x_a - \xi _a \big( \varpi ( \zeta , \boldsymbol{x} ) \big) 
\Big)}}\, = \, \mathrm{cosh} \,  
\psi ( \zeta , \boldsymbol{x} )
\, .
\end{equation}
The worldline is bounded away from the light-cone of $\boldsymbol{x}$
if and only if the right-hand side is bounded for $\zeta 
\to \infty$, i.e., if and only if there exists $\delta > 0$
such that 
\begin{equation}\label{eq:ba2}
\frac{x^0 - \xi ^0 ( \tau ) }{
\sqrt{- \Big(
x^a - \xi ^a ( \tau ) 
\Big) \Big(
x_a - \xi _a ( \tau ) 
\Big)}} \, < \, \delta 
\end{equation}
for $\tau _{\mathrm{min}} < \tau < \tau _0$ with some 
$\tau _0$. To prove the 
first part of the proposition, we assume that this 
condition holds in the chosen inertial system. Under a
Lorentz transformation, $\tilde{x}{}^a = \Lambda ^a{}_b
x^b$, the denominator on the left-hand side of (\ref{eq:ba2})
is unchanged,
\begin{equation}\label{eq:ba3}
\Big(
\tilde{x}{}^a - \tilde{\xi}{}^a ( \tau ) 
\Big) \Big(
\tilde{x}{}_a - \tilde{\xi}{}_a ( \tau ) 
\Big) \, = \, 
\Big(
x^a - \xi ^a ( \tau ) 
\Big) \Big(
x_a - \xi _a ( \tau ) 
\Big) \, ,
\end{equation}
while the numerator changes according to
\begin{equation}\label{eq:ba4}
\tilde{x}{}^0 - \tilde{\xi}{}^0 ( \tau ) \, = \,  
\Lambda ^0{}_0 \big( x^0 - \xi ^0 ( \tau ) \big) +
\Lambda ^0{}_{\mu} \big( x^{\mu} - \xi ^{\mu} ( \tau ) \big) \, . 
\end{equation}
As $\boldsymbol{x} - \boldsymbol{\xi} ( \tau )$ is
timelike and future-pointing, 
\begin{equation}\label{eq:ba5}
\big| x^{\mu} - \xi ^{\mu} ( \tau ) \big|
<
 x^0 - \xi ^0 ( \tau ) 
\end{equation}
for $\mu = 1,2,3$.
As a consequence, (\ref{eq:ba4}) implies that
\begin{equation}\label{eq:ba6}
\tilde{x}{}^0 - \tilde{\xi}{}^0 ( \tau ) 
< 
K  \Big( x^0 - \xi ^0 ( \tau ) \Big) \end{equation}
with some positive constant K. Here we have 
assumed that the Lorentz transformation is orthochronous,
$\Lambda ^0{}_0 >0$. From (\ref{eq:ba2}) we 
find that
\begin{equation}\label{eq:baw6}
\frac{
\tilde{x}{}^0 - \tilde{\xi}{}^0 ( \tau ) 
}{\sqrt{ - \Big(
\tilde{x}{}^a - \tilde{\xi}{}^a ( \tau ) 
\Big) \Big(
\tilde{x}{}_a - \tilde{\xi}{}_a ( \tau ) 
\Big)}} \, < \, K \, \delta 
\end{equation}
for $\tau _{\mathrm{min}} < \tau < \tau _0$ which proves 
that the condition of the worldline being bounded away 
from the light-cone of the chosen event holds in the 
twiddled coordinate system as well. To prove the second 
part of the proposition, we observe that (\ref{eq:psi}) 
implies
\begin{equation}\label{eq:ba7}
\frac{\big|
\vec{x}{} - \vec{\xi} \big( \varpi ( \zeta , \boldsymbol{x} ) \big) 
\big| }{
x^0 - \xi ^0 \big( \varpi ( \zeta , \boldsymbol{x} ) \big) 
} \, = \, \big| \mathrm{tanh} \,  
\psi ( \zeta , \boldsymbol{x} ) \big|
\, .
\end{equation}
Here and in the following, we write $\big| \vec{a} \big| = 
\sqrt{\delta _{\mu \nu} a^{\mu} a^{\nu}}$ for any $\vec{a} =
(a^1,a^2,a^3)$. The worldline is bounded away from
the light-cone of $\boldsymbol{x}$ if and only if the 
right-hand side of (\ref{eq:ba7}) is bounded away from 1
for $\zeta \to \infty$, i.e., if and only if there is a 
$\lambda$ with $0 < \lambda < 1$ such that
\begin{equation}\label{eq:baw7}
\frac{\big|
\vec{x}{} - \vec{\xi} ( \tau ) 
\big| }{
x^0 - \xi ^0 ( \tau )
} \, < \, \lambda
\end{equation}
for $\tau _{\mathrm{min}} < \tau < \tau _0$ with some 
$\tau _0$. Let us assume that this condition holds for 
some particular event $\boldsymbol{x}$. Let $\boldsymbol{y}$ 
be any other event and choose a constant $\mu$ such that 
$\lambda < \mu < 1$. Define
\begin{equation}\label{eq:ba8}
t:= \frac{\mu y^0 - \lambda x^0 - \big| \vec{y} - \vec{x} \big| 
}{\mu - \lambda} \, .
\end{equation}
Then we have, for all $\tau$ such that $\xi ^0 ( \tau) <t$,
\begin{eqnarray}\label{eq:ba9}
\big| \vec{y} - \vec{\xi} ( \tau ) \big| -
\mu \big( y^0 - \xi ^0 ( \tau ) \big) 
&=&
\big| \vec{y} - \vec{x} + \vec{x} - \vec{\xi} ( \tau ) \big| -
\mu \big( y^0 - \xi ^0 ( \tau ) \big) 
\nonumber
\\
&\le&
\big| \vec{y} - \vec{x} \big| + \big| \vec{x} - \vec{\xi} ( \tau ) \big| -
\mu \big( y^0 - \xi ^0 ( \tau ) \big) 
\nonumber
\\
&<&
\big| \vec{y} - \vec{x} \big| + \lambda \big( x^0 - \xi ^0 ( \tau ) \big) -
\mu \big( y^0 - \xi ^0 ( \tau ) \big) 
\nonumber
\\
&<&
\big| \vec{y} - \vec{x} \big| + \lambda  x^0 - 
\mu y^0 + ( \mu - \lambda ) \, t \, = \, 0 
\end{eqnarray}
hence
\begin{equation}\label{eq:ba10}
\frac{\big| \vec{y} - \vec{\xi} ( \tau ) \big|
}{
y^0 - \xi ^0 ( \tau ) }
< \mu 
\end{equation}
for $\tau _{\mathrm{min}} < \tau < \hat{\tau}{} _0$ with some 
$\hat{\tau}{}_0$. This inequality demonstrates that $\boldsymbol{y}$ 
is in the future of the worldline and that the 
worldline is bounded away from the light-cone of $\boldsymbol{y}$ 
as well.
\end{proof}

Because of this result, we may simply say that a worldline is bounded 
away from the past light-cone, without any reference to a specific event 
$\boldsymbol{x}$. 

We now show that the field of a point charge is finite if its
worldline is bounded away from the past light-cone. 
Using the notation of (\ref{eq:chi}) and (\ref{eq:psi}), the 
potential (\ref{eq:Apoint2}) reads
\begin{equation}\label{eq:Apoint4}
\fl
A_a  ( \boldsymbol{x} ) =  
- \, \frac{q}{\ell}  \int _0 ^{\infty} \!
\frac{
\Big( \eta _{a0} \, + \, 
\mathrm{tanh} \, \chi (\zeta , \boldsymbol{x}  ) 
\, \nu ^{\rho} (\zeta , \boldsymbol{x} ) \,  
\eta _{a \rho} \Big) \; J_1 \big( \zeta / \ell \big) \, d \zeta
}{
\mathrm{cosh} \, \psi (\zeta , \boldsymbol{x}  )  \,
\Big( \, 1 \, - \, 
\mathrm{tanh} \, \psi  (\zeta , \boldsymbol{x} )  
\, \mathrm{tanh} \, \chi (\zeta , \boldsymbol{x}  ) \, 
\mu ^{\rho} (\zeta , \boldsymbol{x} )  \, 
\nu _{\rho} (\zeta , \boldsymbol{x} ) \Big) \; \zeta
} 
\; .
\end{equation}
For expressing the field strength tensor at a chosen
event $\boldsymbol{x}$ using the notation of (\ref{eq:chi})
and (\ref{eq:psi}), we may choose the inertial 
coordinate system such that $\dot{\xi}{}^a \big( \tau _R
( \boldsymbol{x} ) \big) = \delta ^a _0$. Then the electric 
and magnetic components of the field strength tensor 
(\ref{eq:Fpoint1}) read, respectively,
\begin{equation}\label{eq:Fe}
\fl
F_{0 \sigma}  (\boldsymbol{x} ) \! =  \! 
\frac{q \, n_{\sigma} ( \boldsymbol{x} )}{2 \ell ^2} \,   
+ \, \frac{q}{\ell ^2} \!
\int _{0} ^{\infty} \!
\frac{
\Big( \mathrm{tanh} \, \psi ( \zeta , \boldsymbol{x} ) \,
\mu _{\sigma}  ( \zeta , \boldsymbol{x} ) -
\mathrm{tanh} \, \chi ( \zeta , \boldsymbol{x} ) \,
 \nu _{\sigma}  ( \zeta , \boldsymbol{x} )
\Big)  J_2 \big( \zeta / \ell \big) \, d \zeta
}{
\Big( \, 1 \, - \, 
\mathrm{tanh} \, \psi  (\zeta , \boldsymbol{x} )  
\, \mathrm{tanh} \, \chi (\zeta , \boldsymbol{x}  ) \, 
\mu ^{\rho} (\zeta , \boldsymbol{x} )  \, 
\nu _{\rho} (\zeta , \boldsymbol{x} ) \Big) \; \zeta
} 
\, , 
\end{equation}
\begin{eqnarray}\label{eq:Fm}
\fl
F_{\rho \sigma}  (\boldsymbol{x} ) \, =  \, 
\\
\nonumber\fl\quad
\frac{q}{\ell ^2} 
\int _{0} ^{\infty}
\frac{
\mathrm{tanh} \, \psi ( \zeta , \boldsymbol{x} ) \,
\mathrm{tanh} \, \chi ( \zeta , \boldsymbol{x} ) \,
\Big( 
\nu _{\sigma}  ( \zeta , \boldsymbol{x} )
\mu _{\rho}  ( \zeta , \boldsymbol{x} ) 
-
\mu _{\sigma}  ( \zeta , \boldsymbol{x} )
\nu _{\rho}  ( \zeta , \boldsymbol{x} ) \Big)
\; J_2 \big( \zeta / \ell \big) \, d \zeta
}{
\Big( \, 1 \, - \, 
\mathrm{tanh} \, \psi  (\zeta , \boldsymbol{x} )  
\, \mathrm{tanh} \, \chi (\zeta , \boldsymbol{x}  ) \, 
\mu ^{\rho} (\zeta , \boldsymbol{x} )  \, 
\nu _{\rho} (\zeta , \boldsymbol{x} ) \Big) 
\; \zeta} 
\, . 
\end{eqnarray}
In a coordinate system with $\dot{\xi}{}^a ( \tau _0 ) 
= \delta ^a _0$ the self-force (\ref{eq:sf2a})
is given by
\begin{eqnarray}\label{eq:sf3}
\fl f^{\textup{s}}_a ( \tau _0 ) \,  = 
\\
\nonumber
\fl -
\frac{q^2 }{\ell ^2} 
\, \eta _{a \sigma}
\int _{0} ^{\infty}
\frac{
\Big( \mathrm{tanh} \, \psi \big( \zeta , \boldsymbol{\xi} (\tau _0 ) \big) \,
\mu ^{\sigma}  \big( \zeta , \boldsymbol{\xi} (\tau _0 ) \big) -
\mathrm{tanh} \, \chi \big( \zeta , \boldsymbol{\xi} (\tau _0 ) \big) \,
 \nu ^{\sigma}  \big( \zeta , \boldsymbol{\xi} (\tau _0 ) \big)
\Big) \; J_2 \big( \zeta / \ell \big) \, d \zeta
}{
\Big( \, 1 \, - \, 
\mathrm{tanh} \, \psi  \big( \zeta , \boldsymbol{\xi} (\tau _0 ) \big)  
\, \mathrm{tanh} \, \chi \big( \zeta , \boldsymbol{\xi} (\tau _0 ) \big) \, 
\mu ^{\rho} \big( \zeta , \boldsymbol{\xi} (\tau _0 ) \big)  \, 
\nu _{\rho} \big( \zeta , \boldsymbol{\xi} (\tau _0 ) \big) \Big) 
\; \zeta} 
\, .
\end{eqnarray}
 
We can now prove the following result.

\begin{proposition}\label{prop:bafc}
If the worldline $\boldsymbol{\xi}$ is bounded away from the 
past light-cone, the integrals on the right-hand sides of 
(\ref{eq:Apoint4}), (\ref{eq:Fe}) and (\ref{eq:Fm}) are absolutely 
convergent for all $\boldsymbol{x} \in \mathbb{R}^4$; the integral
on the right-hand side of (\ref{eq:sf3}) is absolutely convergent
for all $\tau _{\mathrm{min}} < \tau _0 < \tau _{\mathrm{max}}$.
\end{proposition}

\begin{proof} As the integrands in (\ref{eq:Apoint4}), (\ref{eq:Fe}), 
(\ref{eq:Fm}) and (\ref{eq:sf3}) stay finite for $\zeta \to 0$, we only 
have to verify that they fall off sufficiently quickly for $\zeta \to 
\infty$. We first observe that $\mathrm{tanh} \, \alpha$ increases from 
$-1$ to $1$ if $\alpha$ varies from $- \infty$ to $\infty$, and that 
$|\mu ^{\rho} \nu _{\rho} | \le 1$. Therefore, the condition of 
$\psi$ being bounded implies that  $\mathrm{tanh} \, {\psi} \, 
\mathrm{tanh} \, \chi \, \mu ^{\rho} \nu _{\rho}$ is bounded away 
from 1. Thus, in each of the four equations (\ref{eq:Apoint4}), 
(\ref{eq:Fe}), (\ref{eq:Fm}) and(\ref{eq:sf3}) the modulus of the 
integrand is bounded by a term of the form $K \big| J_k 
\big( \zeta / \ell \big) \big| / \zeta$ where $K$ is independent
of $\zeta$ and $k$ is either 1 or 2. As $\big| J_k 
\big( \zeta / \ell \big) \big|$ falls off like $\zeta ^{-1/2}$ for 
$\zeta \to \infty$, this guarantees absolute convergence of the integral. 
\end{proof}
This proposition demonstrates that, in particular, the self-force
is finite for a large class of worldlines. Actually, the requirement 
of the worldline being bounded away from the past light-cone is 
sufficient but not necessary for finiteness of the self-force. The 
following proposition shows that there is another
class of worldlines, including ones which are \emph{not} bounded
away from the past light-cone, for which the self-force is finite.     
\begin{proposition}\label{prop:1d}
Assume that the worldline $\boldsymbol{\xi}$ is confined to a 
two-dimensional timelike plane $P$ in Minkowski spacetime.
Then the integral on the right-hand side of  
(\ref{eq:sf3}) is absolutely convergent for all 
$\tau _{\mathrm{min}} < \tau _0 < \tau _{\mathrm{max}}$. 
\end{proposition}
\begin{proof}
As the worldline is in a timelike plane, the unit vectors
$\nu _{\sigma}  \big( \zeta , \boldsymbol{\xi} ( \tau _0) \big)$ and 
$\mu _{\sigma}  \big( \zeta , \boldsymbol{\xi} ( \tau _0) \big)$ must 
be constant and equal or opposite, $\nu _{\sigma} = \pm \mu _{\sigma}$.
Then  (\ref{eq:sf3}) simplifies to
\begin{equation}\label{eq:sf4}
\fl f^{\textup{s}}_a ( \tau _0 ) \,  = \, 
- \, \frac{q^2}{\ell ^2} 
\, \eta _{a \sigma} \, \mu ^{\sigma}
\int _{0} ^{\infty}
\,
\mathrm{tanh} \Big( \psi \big( \zeta , \boldsymbol{\xi} (\tau _0 ) \big) 
\mp \chi \big( \zeta , \boldsymbol{\xi} (\tau _0 ) \big) \Big)
\,
\frac{J_2 ( \zeta / \ell )  \, d \zeta}{\zeta}  
\, .
\end{equation}
As $\big| \mathrm{tanh} (\alpha) \big| \le 1$, the modulus of
the integrand in (\ref{eq:sf4}) is bounded by $\big| J_2 
\big( \zeta / \ell \big) \big| / \zeta$. As $\big| J_2 
\big( \zeta / \ell \big) \big|$ falls off like $\zeta ^{-1/2}$ 
for $\zeta \to \infty$, this guarantees absolute convergence 
of the integral. 
\end{proof}


\subsection*{Example 3: A worldline with diverging self-force integral}

{}From Propositions~\ref{prop:bafc} and \ref{prop:1d} we know that 
the self-force is finite if the particle's worldline  is bounded 
away from the light-cone or if it is contained in a timelike 
plane. Actually, the proof of Proposition~\ref{prop:1d} can be 
generalized to the case that the worldline, rather than being 
confined to $P$, approaches $P$ sufficiently quickly for 
$\tau \to \tau _{\mathrm{min}}$.  This leaves only a class 
of rather contrived motions for which the self-force integral 
could diverge: The worldline must approach the light-cone for 
$\tau \to \tau _{\mathrm{min}}$ with a sufficiently large 
tangential velocity component. In this section we present such an 
example for which the self-force integral, indeed, diverges at
one point.

We find it convenient to give the worldline in terms of a 
\emph{past-oriented} curve parameter $\gamma$ which is 
\emph{not} proper time, 
\begin{equation}\label{eq:os}
\xi ^a ( \gamma ) \, = \, \gamma \Big\{
-\sqrt{1+s( \gamma )} \delta _0 ^a + \sqrt{s( \gamma )} 
\Big( \mathrm{cos} \big( \beta ( \gamma ) \big) \delta _1^a
+ \mathrm{sin} \big( \beta ( \gamma ) \big) \delta _2 ^a \Big)
\Big\}
\end{equation}
where
\begin{equation}\label{eq:s}
s ( \gamma ) \, = \, 
\frac{\gamma}{\ell \, \mathrm{tanh} \, ( \gamma / \ell ) }
\end{equation}
and 
\begin{equation}\label{eq:beta}
\fl
\beta ( \gamma ) \, = \, 
\frac{1}{2} 
\left( 
\sqrt{\pi} C_F \big( \sqrt{2s(\gamma ) / \pi} \big)
-
\sqrt{\pi} S_F \big( \sqrt{2s(\gamma ) / \pi} \big)
-
\frac{
\mathrm{sin} \big( s ( \gamma ) \big) 
+ \mathrm{cos} \big( s ( \gamma ) \big)
}{
\sqrt{2s(\gamma )}}
\right)
\, .
\end{equation}
Here $C_F$ and $S_F$ are the Fresnel-C and Fresnel-S 
functions. Note that
\begin{equation}\label{eq:betap}
\beta ' ( \gamma ) \, = \, 
\frac{
s'( \gamma ) \mathrm{sin} \big( s ( \gamma ) + \pi / 4 \big)
}{
4 \, s(\gamma ) ^{3/2}}
\, .
\end{equation}
$\boldsymbol{\xi}$ is an analytic timelike curve that 
approaches the past light-cone with an oscillatory 
tangential velocity component, see Fig.~\ref{fig:os}.
The curve is, indeed, everywhere timelike as can be 
seen from
\begin{eqnarray}\label{eq:xitime}
\eta _{ab} \frac{d \xi ^a ( \gamma )}{d\gamma}
\frac{d \xi ^b ( \gamma )}{d\gamma}
&=&
-1 + \frac{\gamma ^2 s'( \gamma )^2}{4 s( \gamma )
\big( 1 + s ( \gamma ) \big)}
+
\frac{\gamma ^2 s'( \gamma )^2}{16 s( \gamma )^2}
\, \mathrm{sin} ^2 \big(s( \gamma ) +\pi /4 \big)
\nonumber
\\
&<& -1 + \frac{1}{4} + \frac{1}{16} = \frac{-11}{16}
\, .
\end{eqnarray}
It can be shown that the 4-acceleration of $\boldsymbol{\xi}$ is 
bounded and that its future is all of Minkowski
spacetime.
 
\begin{figure}
\setlength{\unitlength}{1cm}
     \psfrag{x}{{}} 
    \psfrag{y}{{}} 
     \psfrag{c}{{}} 
     \psfrag{d}{{}} 
    \psfrag{a}{{\hspace{-2em}\raisebox{-0.5em}{\small $0.99$}}}
    \psfrag{b}{{\hspace{-1em}\raisebox{-0.5em}{\small $1$}}} 
\begin{picture}(17,6.4)
\put(1.3,0){\epsfig{figure=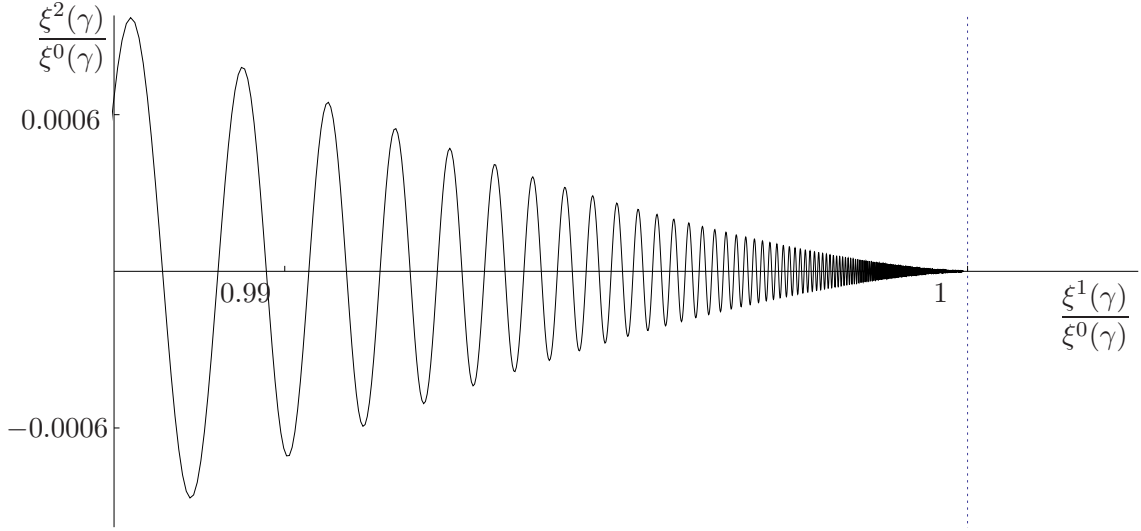,width=14cm}}
\put(0.2,5.3){\small$0.0006$}
\put(0,1.2){\small$-0.0006$}
\put(0.4,6.4){\small$\displaystyle\frac{\xi^2(\gamma )}{\xi^0(\gamma )}$}
\put(14,2.7){\small$\displaystyle\frac{\xi^1(\gamma )}{\xi^0(\gamma )}$}
\end{picture}
\caption{The worldline \protect{(\ref{eq:os})} approaches 
the past light-cone (dotted circle) with an oscillatory 
tangential velocity component}
\label{fig:os}
\end{figure}

We show that for this worldline $\boldsymbol{\xi}$ the 
electromagnetic field (\ref{eq:Fpoint1}) diverges on a 
hyperplane implying that the self-force becomes infinite at one instant.
To that end we have to rewrite the integral in (\ref{eq:Fpoint1}) as
an integral over $\gamma$ and to investigate the behavior of the 
integrand for $\gamma \to \infty$. 
We first observe that, if $\gamma$ tends to $\infty$,
\begin{eqnarray}\label{eq:slim}
s( \gamma ) &=& \frac{\gamma}{\ell} + O \big( \gamma ^{-n} \big)
\, , \quad
s'( \gamma ) = \frac{\gamma}{\ell} + O \big( \gamma ^{-n} \big)
\, , \nonumber\\
\beta '( \gamma ) &=& \frac{\sqrt{\ell} \, \mathrm{sin} \Big(
\frac{\gamma}{ \ell} + \frac{\pi}{4} \Big)}{\gamma ^{3/2}}
 + O \big( \gamma ^{-n} \big)
\end{eqnarray}
for all $n \in \mathbb{N}$. Moreover, from the standard asymptotic
formulas for the Fresnel functions we find
\begin{equation}\label{eq:betalim}
\mathrm{cos} \big( \beta ( \gamma ) \big) = 1 +  O \big( \gamma ^{-3} \big)
\, , \qquad
\mathrm{sin} \big( \beta ( \gamma ) \big) = O \big( \gamma ^{-3/2} \big)
\, .
\end{equation}
With the help of these formulas we find that the parameter $\zeta$ 
which is used as the integration variable in (\ref{eq:Fpoint1}) is 
related to our curve parameter $\gamma$ by
\begin{equation}\label{eq:zetagamma}
\zeta ^2 = - \big( x^a - \xi ^a (\gamma ) \big) 
\big( x_a - \xi _a (\gamma ) \big) 
= \gamma ^2 + 2 \frac{\gamma ^{3/2}}{\sqrt{\ell}} \big( x^0+x^1 \big)
+ O \big( \gamma ^0 \big) \, ,
\end{equation}
hence
\begin{equation}\label{eq:zetalimit}
\zeta  = \gamma + \sqrt{ \frac{\gamma}{\ell}} \big( x^0+x^1 \big)
- \frac{(x^0+x^1)^2}{\ell} + O \big( \gamma ^{-1/2} \big) \, 
\, .
\end{equation}
Inserting this expression into the well-known asymptotic formula for 
the Bessel function $J_2$ yields
\begin{eqnarray}\label{eq:J2limit}
J_2 \Big( \frac{\zeta}{\ell} \Big) = - \sqrt{\frac{2 \ell}{\pi \zeta}}
\, \mathrm{sin} \Big( \frac{\zeta}{\ell} + \frac{\pi}{4} \Big) + 
O \big( \zeta ^{-3/2} \big) 
\\
\nonumber
=
- \sqrt{\frac{2 \ell}{\pi \gamma}}
\mathrm{sin} \Big( \frac{\gamma}{\ell} + \frac{\pi}{4} 
+ \frac{\sqrt{\gamma} (x^0+x^1)}{\ell ^{3/2}} 
- \frac{(x^0+x^1)^2}{\ell ^2}
\Big) + 
O \big( \gamma ^{-1} \big) 
\, .
\end{eqnarray}
After these preparations, we are ready to evaluate the integral in
(\ref{eq:Fpoint1}). If we use $\gamma$ as the integration variable,
writing $\gamma _R ( \boldsymbol{x} )$ for the parameter value that
corresponds to $\tau _R ( \boldsymbol{x} )$ and thus to $\zeta =0$, 
this integral reads 
\begin{eqnarray}\label{eq:Iab}
I_{ab}
&=&
\int _{\gamma _R ( \boldsymbol{x} )} ^{\infty}
\frac{
\left( \big( x_b -\xi _b ( \gamma  ) \big) \frac{d \xi _a ( \gamma)}{d \gamma}
-
\big( x_a -\xi _a ( \gamma  ) \big) \frac{d \xi _b ( \gamma)}{d \gamma} \right)
}{
\frac{d \xi ^c ( \gamma )}{d \gamma} 
\big( x_c -\xi _c ( \gamma  ) \big)
} 
\frac{
J_2 (\zeta / \ell )}{\zeta} \, \frac{d \zeta}{d \gamma} \, d \gamma
\\
&=&
\int _{\gamma _R ( \boldsymbol{x} )} ^{\infty}
\frac{
\left( \big( x_b -\xi _b ( \gamma  ) \big) \frac{d \xi _a ( \gamma)}{d \gamma}
-
\big( x_a -\xi _a ( \gamma  ) \big) \frac{d \xi _b ( \gamma)}{d \gamma} \right)
}{
\zeta \, \frac{d \zeta}{d \gamma} } 
\,
\frac{
J_2 (\zeta / \ell )}{\zeta} \, \frac{d \zeta}{d \gamma} \, d \gamma
\nonumber
\\
&=&
\int _{\gamma _R ( \boldsymbol{x} )} ^{\infty}
\left( \big( x_b -\xi _b ( \gamma  ) \big) \frac{d \xi _a ( \gamma)}{d \gamma}
-
\big( x_a -\xi _a ( \gamma  ) \big) \frac{d \xi _b ( \gamma)}{d \gamma} \right)
\,
\frac{J_2 (\zeta / \ell )}{\zeta ^2} \, d \gamma
\, .
\nonumber
\end{eqnarray}
We evaluate this equation for $a=2$ and $b=0$. As
\begin{equation}\label{eq:02}
\fl \big( x_0-\xi _0 ( \gamma) \big) \frac{d \xi _2 ( \gamma )}{d \gamma}
-
\big( x_2-\xi _2 ( \gamma) \big) \frac{d \xi _0 ( \gamma )}{d \gamma}
=
\frac{- \gamma ^{3/2}}{4 \sqrt{\ell}} 
\, \mathrm{sin} \Big( \frac{\gamma}{\ell} + \frac{\pi}{4} \Big)
+ O \big( \gamma ^{1/2} \big) \, ,
\end{equation}
we find with our asymptotic formula for $J_2 (\zeta / \ell )$ from above
\begin{eqnarray}\label{eq:I20}
\fl I_{20}
&=&
\int _{\gamma _R ( \boldsymbol{x} )} ^{\infty}
\left\{
\frac{
\ell \, \mathrm{sin} 
\Big( \frac{\gamma}{\ell} + \frac{\pi}{4} \Big)
}{
2 \sqrt{2 \pi} \, \gamma
}
\, 
\mathrm{sin} \Big( \frac{\gamma}{\ell} + \frac{\pi}{4} 
+ \frac{\sqrt{\gamma} (x^0+x^1)}{\ell ^{3/2}}  
- \frac{(x^0+x^1)^2}{\ell ^2} \Big) 
+ O \big( \gamma ^{-3/2} \big)
\right\}
\, d \gamma
\nonumber
\\
\fl &=&
\frac{\ell}{4 \sqrt{2 \pi}}
\int _{\gamma _R ( \boldsymbol{x} )} ^{\infty}
\!\!
\left\{
\mathrm{cos} 
\Big( \frac{\sqrt{\gamma} (x^0+x^1)}{\ell ^{3/2}}  
- \frac{(x^0+x^1)^2}{\ell ^2} \Big) 
\right.
\\ \fl &&
\left. 
\hspace{7em} -
\mathrm{cos} 
\Big( \frac{2 \gamma}{\ell} + \frac{\pi}{2} +
\frac{\sqrt{\gamma} (x^0+x^1)}{\ell ^{3/2}}  
- \frac{(x^0+x^1)^2}{\ell ^2} \Big) 
\right\} \frac{d \gamma}{\gamma}
+ \, \dots
\nonumber
\end{eqnarray}
Here and in the following, the ellipses indicate a term that 
is finite for all $\boldsymbol{x}$. The integral over the 
second term is finite for all $\boldsymbol{x}$. If we decompose 
the remaining integral into an integration from $\gamma _R 
( \boldsymbol{x})$ to $\ell$ and an integration from $\ell$ 
to infinity we find
\begin{eqnarray}\label{eq:I20b}
\fl I_{20}
&=&
\frac{\ell}{4 \sqrt{2 \pi}}
\int _{\ell} ^{\infty}
\mathrm{cos} 
\Big( \frac{\sqrt{\gamma} (x^0+x^1)}{\ell ^{3/2}}  
- \frac{(x^0+x^1)^2}{\ell ^2} \Big) 
\frac{d \gamma}{\gamma}
+ \, \dots
\nonumber
\\ \fl
&=&
\frac{\ell}{4 \sqrt{2 \pi}}
\left\{-2 \, \mathrm{cos} \left( \frac{(x^0+x^1)^2}{\ell ^2} \right)
\, \mathrm{Ci} \left( \frac{|x^0+x^1|^{3/2}}{\ell ^{3/2}} \right)
\right.
\\
\nonumber
\fl &&\hspace{7em}
\left.
-
\, \mathrm{sin} \left( \frac{(x^0+x^1)^2}{\ell ^2} \right)
\left(
\pi - 2 \, \mathrm{Si} \left( \frac{|x^0+x^1|^{3/2}}{\ell ^{3/2}} \right)
\right)
\right\} \, + \, \dots
\end{eqnarray}
Here Ci and Si denote the cosine integral and the sine integral,
respectively. As the cosine integral diverges logarithmically
if its argument approaches zero, we have found that the 
(20)-component of the electromagnetic field according to 
(\ref{eq:Fpoint1}) is given by
\begin{equation}\label{eq:I20c}
F_{20} ( \boldsymbol{x} ) \, = \, - \, \frac{3 q}{4 \sqrt{2 \pi} \ell}
\, \mathrm{log} \left( \frac{| x^0+x^1 |}{\ell} \right)
\, + \, \dots \, ,
\end{equation}
i.e., that this field component diverges on the lightlike hyperplane
$x^0+x^1=0$. This may be interpreted as a shock front propagating at 
the speed of light. The same divergence is found, by a completely 
analogous calculation, for the component $F_{21}=-F_{12}$,
\begin{equation}\label{eq:I21c}
F_{21}  ( \boldsymbol{x} ) \, = \, - \, \frac{3 q}{4 \sqrt{2 \pi} \ell}
\, \mathrm{log} \left( \frac{| x^0+x^1 |}{\ell} \right)
\, + \, \dots \, ,
\end{equation}
while all other
components are finite everywhere. As a consequence, the self-force
becomes infinite at the instant when the charged particle crosses
the hypersurface $x^0+x^1=0$ which happens at the origin of the 
coordinate system. This means that at this instant an infinite
external \cRed{Minkowski} force is necessary to keep the particle on its prescribed
worldline.

Note that the divergence is logarithmic in the neighborhood of a
lightlike hyperplane and thus rather mild, since any timelike
worldline crosses it at most once. It is true that  
a charged particle would experience an infinite \cRed{relativistic} Lorentz
force at the instant when it crosses the hypersurface $x^0+x^1 =0$.
However, as $\int^y\mathrm{log} |x| \, dx$ is finite-valued and 
continuous at $y=0$, the particle's velocity would still be 
finite-valued and continuous, i.e., the particle's worldline
would still be a $C^1$ curve. 

\section{The Abraham-Lorentz-Dirac limit}\label{sec:LD}
 
In the standard Maxwell-Lorentz theory with point charges, i.e., for $\ell \to 0$,
  the self-force becomes infinite in (\ref{eq:sf1}). Dirac's solution
  to give a meaning to the equation of motion (\ref{eq:eom}) in this
  case was to assume that the inertial mass became negative infinite
  in order to cancel the infinite contribution from the self-force.
After this cancellation, one ends up with the
Abraham-Lorentz-Dirac equation which involves a
renormalized (or dressed) mass which is positive and finite. It is
interesting to see how the Abraham-Lorentz-Dirac equation
is reproduced from the Bopp-Podolsky theory in the limit $\ell \to
0$. To that end we substitute in (\ref{eq:sf2a}) the integration
variable $\zeta = \ell \sigma$,
\begin{equation}\label{eq:sf2b}
f^{\textup{s}}_a ( \tau _0 ) \,  = \, 
\frac{q^2 \dot{\xi}{}^b ( \tau _0 )}{\ell } 
\int _{0} ^{\infty}
\frac{\partial W_{ab}}{\partial \zeta}
\big( \ell \sigma, \boldsymbol{\xi} ( \tau _0 ) \big) 
\frac{d^2 \chi (\sigma )}{d \sigma ^2}
\, d \sigma
\end{equation}
where
\begin{equation}\label{eq:chi2}
\chi ( \sigma ) = \int _{\sigma} ^{\infty} \left(
\int _{\sigma '} ^{\infty} 
\frac{J_1 ( \sigma '' )}{\sigma ''} d \sigma '' \right) 
d \sigma '
\, .
\end{equation}
With
\begin{equation}\label{eq:chipp}
\chi (0 ) = 1 \, , \qquad \chi ' ( 0 ) = -1
\end{equation}
two times integrating by parts yields
\begin{eqnarray}\label{eq:sf2c}
f^{\textup{s}}_a ( \tau _0 ) \,  &=& \, 
- \, \frac{q^2}{2 \ell} \ddot{\xi}{}_a ( \tau _0 )
\, + \, 
\frac{2}{3} \, \Big( \stackrel{\bm{...}}{\xi}_a \!( \tau _0 )\,
+ \dot{\xi}{}_a( \tau _0 ) \,
\dot{\xi}{}^b ( \tau _0 ) \stackrel{\bm{...}}{\xi}_b ( \tau _0 ) \Big) 
\\
\nonumber
&&+ \, \ell q^2 \dot{\xi}{}^b ( \tau _0 ) 
\int _{0} ^{\infty}
\frac{\partial ^3 W_{ab}}{\partial \zeta ^3}
\big( \ell \sigma, \boldsymbol{\xi} ( \tau _0 ) \big) 
\, \chi (\sigma ) \, d \sigma
\, .
\end{eqnarray}
The first term diverges for $\ell \to 0$.  
Following Dirac's idea of mass renormalization, 
the parameter $m$ must depend on $\ell$ and become negative infinite 
such that the ``dressed mass''
\begin{equation}\label{eq:ren}
\hat{m} \, = \, 
\lim_{\ell\to0}
\Big( m(\ell) + \frac{q^2}{2 \ell} \Big) 
\end{equation}
remains finite and positive. In this limit,
the 
equation of motion reads
\begin{eqnarray}\label{eq:LD}
\hat{m} \, \ddot{\xi}_a ( \tau  ) 
&=&
\frac{2q^2}{3} \Big( \stackrel{\bm{...}}{\xi}_a ( \tau  ) + 
\dot{\xi}{}_a ( \tau  ) \dot{\xi}^b ( \tau  )
\stackrel{\bm{...}}{\xi}_b ( \tau  ) \Big)
\, + \, f^{\textup{e}} _a ( \tau ) 
\\
\nonumber
&&+ \, \lim_{\ell\to0}
\left( \ell q^2 \dot{\xi}{}^b ( \tau _0 ) 
\int _{0} ^{\infty}
\frac{\partial ^3 W_{ab}}{\partial \zeta ^3}
\big( \ell \sigma, \boldsymbol{\xi} ( \tau _0 ) \big) 
\, \chi (\sigma ) \, d \sigma
\right) \, .
\end{eqnarray}

If the integral is bounded, the last term vanishes
and we get the Abraham-Lorentz-Dirac equation. From (\ref{eq:chi2})
we find, with the help of the well-known asymptotic formula
for the Bessel function $J_1$, that $\chi (\sigma ) = 
O ( \sigma ^{-3/2} )$ for $\sigma \to \infty$.
So the integral in (\ref{eq:LD}) is certainly bounded if 
$\partial ^3 W_{ab} / \partial \zeta ^3$ is bounded for 
$\zeta \to \infty$. 
A sufficient (but not necessary) condition 
is that $\xi$ is bounded away from the light-cone and that
all components $\dot{\xi}^a$, $\ddot{\xi}^a$, 
$\stackrel{\bm{...}}{\xi^a}$, and $\stackrel{\bm{....}}{\xi^a}$ are bounded.

\section{Conclusions}\label{sec:con}

In this paper we have demonstrated that in the Bopp-Podolsky theory
\cRed{a self-force of a charged point particle can be defined} by an
integral that is absolutely convergent for a large class of worldlines
on Minkowski spacetime. We have also provided a (contrived) example
where \cRed{its} electromagnetic field diverges on a lightlike plane,
so the self-force diverges at one point of the worldline. However,
even in this case the electromagnetic field is a locally integrable
function (i.e., a regular distribution), yielding a $C^1$ solution to
the equation of motion (\ref{eq:eom}).  This is to be contrasted with
the standard \cRed{Maxwell-Lorentz theory} in vacuo where the
self-field of a charged point particle is infinite at every point of
the particle's worldline, and the singularity is so bad that the field
energy in an arbitrarily small ball around the charge is infinite
\cRed{necessitating classical} mass renormalization. In the
Bopp-Podolsky theory there is no need for mass renormalization; the
equation of motion (\ref{eq:eom}) is an integro-differential equation
making sense with a finite inertial mass $m$.
  
It should also be emphasized that in the Born-Infeld theory,
which is to be viewed as a natural rival to the Bopp-Podolsky
theory, virtually nothing is known about finiteness of 
the self-force of an accelerated particle whose worldline 
may approach the light-cone. So it seems fair to say that,
at least in view of the motion of charged point particles,
the Bopp-Podolsky theory is in a more promising state.
    
Some questions remain open. 
\cRed{It is clearly important to establish a strategy for solving
  (\ref{eq:eom}). This may include formulating a consistent Cauchy
  problem that may include either background or dynamic fields
  satisfying the Bopp-Podolsky field equations coupled to the particle.}
Furthermore, it would be desirable 
to have a proof that \emph{all} worldlines solving (\ref{eq:eom}) are
at least $C^1$. \cRed{It is also} important  
to demonstrate that the equation of  motion (\ref{eq:eom})
with vanishing external \cRed{Minkowski} force is free of run-away solutions. 
Partial results in this direction have been found by  
\citeasnoun{FrenkelSantos1999}, but a general proof is still missing.
If these questions can be satisfactorily resolved, the Bopp-Podolsky
theory offers a physically acceptable framework in which to \cRed{explore}
the classical 
electromagnetic back-reaction on charged point particles.

\section*{Acknowledgements}
VP was financially supported by the Deutsche 
Forschungsgemeinschaft, Grant LA905/10-1 and by the 
German-Israeli-Foundation, Grant 1078/2009 during
part of this work. Moreover,
VP gratefully acknowledges support from the Deutsche 
Forschungsgemeinschaft within the Research Training
Group 1620 ``Models of Gravity''. The authors JG and 
RWT are grateful for support provided by STFC (Cockcroft 
Institute ST/G008248/1) and EPSRC (Alpha-X project EP/J018171/1). 
We also thank Dr. David Burton for useful discussions.
There are no additional data files for this article.

\section*{Appendix}

The formulation of a stress-energy-momentum tensor for the theory 
discussed in this paper appears to have had a chequered history.
\citeasnoun{Bopp1940} writes down a stress-energy-momentum tensor
that is obviously based on the decomposition (\ref{eq:Acon}) of 
the potential, but no derivation is given. \citeasnoun{Podolsky1942}
also writes down a stress-energy-momentum tensor with a promise to 
derive it in \cite{PodolskyKikuchi1944} from arguments based on a 
canonical approach. In our view this did not succeed and the 
further derivation in \cite{PodolskySchwed1948} lacked transparency. 
To our knowledge there has been no subsequent attempt to derive any 
stress-energy-momentum tensor appropriate to the theory under discussion.
In view of these comments it may be of value to put the matter into 
a modern perspective by offering a derivation based on metric variations 
of the Bopp-Podolsky action. This requires formulating the theory
on a \emph{curved} spacetime manifold.

The natural tools for this purpose exploit the exterior calculus of 
differential forms  using properties of the Hodge map $\star$ associated 
with the spacetime metric $g$ and the nilpotency of the exterior derivative 
$d$. For background material on exterior calculus we refer to 
\citeasnoun{Straumann1984} whose sign and factor conventions we 
adopt. The Bopp-Podolsky action reads 
\begin{equation}\label{eq:SBAg}
S[A,J,g]=\int_{\Man} \Lambda = 
\int_{\Man}  \big( \Lambda _{\mathrm{EM}} - A \wedge J \big)
\end{equation} 
with
\begin{equation}\label{eq:Lambda}
\Lambda _{\mathrm{EM}} =  
\frac{1}{8 \pi} F \wedge \star F - 
\frac{\ell ^2}{8 \pi} G \wedge \star G 
\end{equation}
where $F=d A$ and $ G= d\star F$ \cRed{are smooth}. Here we assume
that $J$ is a prescribed (non-dynamical) \cRed{smooth} current 3-form that satisfies
the conservation law $dJ=0$ on $\Man$.  In the body of the paper we have
restricted ourselves to the case of a flat metric and used
inertial coordinates. Then the action (\ref{eq:SBAg}) reduces to
(\ref{eq:Lag2}) where the current 4-vector $j = j^a \partial _a$ is
related to the current 3-form by $g(j, \, \cdot \, ) = \star J$.

A direct route to the symmetric dynamical (Hilbert)
stress-energy-momentum tensor, associated with
$\Lambda_{\mathrm{EM}}$, is obtained by making compact variations of
the metric tensor in $\Lambda_{\mathrm{EM}}$.  Such variations can be induced
by making independent variations $\dot e^a$ in a local $g-$orthonormal
coframe $\{e^a\},\,a=0,1,2,3$ since in such a basis $g=\eta_{ab}
e^a\otimes e^b$. Such variations give rise to a set of 3-forms
$\tau_a$ defined by
\begin{equation}\label{eq:deftauc}
\int_{\Man} \dot{\Lambda}_{\mathrm{EM}}  
= \int_{\Man} \dot e^a \wedge \tau_a
\end{equation}
and a stress-energy-momentum tensor 
$T=T_{ab} e^a \otimes e^b$ 
with components  
\begin{equation}\label{eq:Tab}
T_{ab}=\eta_{bc}\star(\tau_a \wedge e^c  ) \, .
\end{equation}
The covariant divergence of $T$ then follows as 
\begin{equation}\label{eq:divTDtau}
\nabla\cdot T= ( \star^{-1} D \tau_a ) e^a
\end{equation}
where $D$ denotes the covariant exterior derivative, see
e.g. \citeasnoun{BennTucker1988}.
Since $T$ is symmetric
\begin{equation}
D\tau_a =  d\tau_a  - i_a d e^b \wedge \tau_b
\ .
\label{AB_def_Dtau}
\end{equation}

If we make compact variations of the potential $A$ in $S$, rather
than of the metric, we \cRed{obtain} the field equation of the 
Bopp-Podolsky theory. We can derive the $\tau _a$ and 
the field equation in one go if we allow for partial variations 
of the potential and of the metric simultaneously. (Note
that the current $J$ is assumed to be given and is
kept fixed during the variation.) Then the total variation 
of the Lagrangian 4-form is written
\begin{equation}\label{eq:dotLambda}
{\Lambda{\boldsymbol{\dot{}}}} = 
\frac{1}{8 \pi} \big( F \wedge \star F \big){\boldsymbol{\dot{}}}
- \frac{\ell ^2}{8 \pi} \big( G \wedge \star G \big){\boldsymbol{\dot{}}}
- {A}{\,\boldsymbol{\dot{}}} \wedge J \, .
\end{equation}
For calculating the extremum of $\Lambda$ we use the formula
\cite{DereliGratusTucker2007}
\begin{equation}\label{eq:DGT}
(\star\Psi)\,{\boldsymbol{\dot{}}}=
e_c{\!\boldsymbol{\dot{}}} \wedge i^c(\star \Psi)
- \star( e_c{\boldsymbol{\!\dot{}}} \wedge i^c\Psi )
+ \star\Psi\,{\boldsymbol{\dot{}}}
\end{equation}
that holds for any $p-$form $\Psi$, where $i_c$ denotes the 
contraction operator (or interior derivative) with respect to the 
vector field $X_c$ defined by $e^a (X_c) = \delta ^a_c$.
Moreover, we use standard rules of exterior calculus, such as 
$ \Psi\wedge \star \Phi =   \Phi\wedge\star \Psi$
for any $p-$forms $\Psi,\Phi$, the graded derivative property
and nilpotency of $d$, and the commutativity of $d$ with 
the variations. Thus
\begin{equation}\label{eq:dot1}
\big( F \wedge \star F \big){\,\boldsymbol{\dot{}}}=
2 \, d \big( {A}{\,\boldsymbol{\dot{}}} \wedge \star F \big)
+ 2 \, {A}{\,\boldsymbol{\dot{}}} \wedge d \star F +
{e}_c{\!\boldsymbol{\dot{}}}
\wedge\big( F \wedge i^c \star F - i^c F \wedge F \big)
\end{equation}
and
\begin{eqnarray}\label{eq:dot2}
\fl
\big( G \wedge \star G \big){\,\boldsymbol{\dot{}}} 
 &=&
2 \, d \left( \Big( {e}_c{\!\boldsymbol{\dot{}}} \wedge i^c \star F
- \star ( {e}_c{\!\boldsymbol{\dot{}}} \wedge i^c F ) + \star
{F}{\,\boldsymbol{\dot{}}} \, \Big)
\wedge \star G \right)
\nonumber
\\
\fl&&
- 2 \, {A}{\,\boldsymbol{\dot{}}} \wedge d \star d \star G
- {e}_c{\!\boldsymbol{\dot{}}} \wedge \Big( 2 \, i^c \star F \wedge d \star G
- 2 \, i^c F \wedge \star d \star G
\nonumber
\\\
\fl&&\qquad
+ G \wedge i^c \star G
+ i^c G \wedge \star G \Big) \, .
\end{eqnarray}
Inserting (\ref{eq:dot1}) and (\ref{eq:dot2}) into
(\ref{eq:dotLambda}) yields
\begin{equation}\label{eq:dotLambda2}
{\Lambda}{\,\boldsymbol{\dot{}}} = d \Phi + {e}{}_c{\!\boldsymbol{\dot{}}} \wedge \tau_c
+ {A}{\,\boldsymbol{\dot{}}} \wedge \, \frac{1}{4 \pi} \, dH
- {A}{\,\boldsymbol{\dot{}}} \wedge J
\end{equation}
where
\begin{equation}\label{eq:Phi}
4 \pi \Phi =
{A}{\,\boldsymbol{\dot{}}} \wedge H -
\ell ^2 \Big( {e}{}_c{\!\boldsymbol{\dot{}}} \wedge i^c \star F
- \star ( {e}{}_c{\!\boldsymbol{\dot{}}} \wedge i^c F ) + \star
{F}{\,\boldsymbol{\dot{}}} \Big)
\wedge \star G
\, ,
\end{equation}
\begin{equation}\label{eq:H}
H = \star F + \ell ^2 \star d \star d \star F
\end{equation}
and
\begin{equation}\label{eq:tauc}
\fl
8\pi \tau _c = 
F \wedge i_c\star F - i_c F\wedge \star F 
+ \ell ^2
\Big( G \wedge i_c\star G + i_c  G \wedge \star G
-2 i_c F \wedge\star d\star G 
+2 i_c \star F \wedge d\star G  \Big) \, .
\end{equation}

The field equations 
are determined by requiring that the action is stationary
for partial variations of the potential only (i.e., $\dot{e}{}^c = 0$) that
are compactly supported $\big($i.e $\int_{\Man} d \Phi = 0 \big)$:
\begin{equation}\label{eq:dHJ}
dH = 4 \pi J  \, .
\end{equation}
Similarly for partial variations with $\dot{A}=0$, 
the $\tau _c$ from (\ref{eq:tauc}) give, via 
(\ref{eq:Tab}), the dynamical
stress-energy-momentum tensor of the Bopp-Podolsky theory
which is automatically symmetric.
A lengthy but routine calculation of $D\tau_a$ then shows, with the 
aid of the field equations  (\ref{eq:dHJ}), that the divergence 
(\ref{eq:divTDtau}) yields the \cRed{relativistic} Lorentz force,
\begin{equation}\label{eq:divT}
\nabla \cdot T
= F \big( \, \cdot \, , j \big) 
\, .
\end{equation}
A less lengthy approach to derive (\ref{eq:divT}) can be based on the
use of a one-parameter family of diffeomorphisms on the spacetime
\cRed{domain $\Man$} generated by any compactly supported vector field
$X$. Then in terms of the Lie derivative $L_X$
\begin{equation}
\label{eq:dotLambdaEM1}
L_X \Lambda _{\mathrm{EM}} = d \Phi_X + L_X e ^c \wedge \tau _c
+ \frac{1}{4\pi}
L_X A \wedge dH
\end{equation}
with the same $\tau_c$ and the same $H$ as in 
(\ref{eq:dotLambda2}) and 
\begin{equation}\label{eq:PhiX}
4 \pi \Phi_X =  
L_X{A} \wedge H - 
\ell ^2 \Big( L_X{e}{}^c \wedge i_c \star F 
- \star ( L_X{e}{}^c \wedge i_cF ) + \star L_X{F} \Big)
\wedge \star G 
\,.
\end{equation}
For the derivation of (\ref{eq:dotLambdaEM1}) we have used 
the identity
\begin{equation}\label{eq:LXstar}
L_X(\star\Psi)=
L_X e^c \wedge i_c(\star \Psi) 
- \star( L_X e^c \wedge i_c\Psi ) 
+ \star L_X\Psi
\end{equation}
which holds for any $p$-form $\Psi$ on $M$. With the 
relations $L_X = i_Xd+di_X$ and $d \Lambda _{\mathrm{EM}}
=0$ (since $\Lambda _{\mathrm{EM}}$ is a 4-form on spacetime), 
(\ref{eq:dotLambdaEM1}) yields the identity
\begin{equation}
d \,\alpha_X = \beta_X
\label{eq:dalpha,beta}
\end{equation}
where
\begin{equation}
\alpha_X
=
i_X \Lambda _{\mathrm{EM}} 
- i_X e ^c \wedge \tau _c
- i_X A \wedge \frac{1}{4\pi}dH
- \Phi_X
\label{eq:def_alpha}
\end{equation}
and
\begin{equation}
\beta_X
=
i_X d e ^c \wedge \tau _c
-i_X e ^c \wedge d \tau _c
+ i_X F \wedge \frac{1}{4\pi}dH
\,.
\label{eq:def_beta}
\end{equation}
Integrating (\ref{eq:dalpha,beta}) over an open domain $U$ containing the
support of $X$ and using Stokes's theorem yields 
\begin{equation*}
\int_U\beta_X=\int_U
d\alpha_X=\int_{\partial U} \alpha_X =0
\end{equation*}
since $\alpha_X|_{\partial U}=0$. Furthermore since $\beta_X$ has the
linearity property $\beta_{fX}=f\,\beta_X$ for any smooth function
$f$, it follows\footnote{
Note however since $\alpha_{fX}\ne f\,\alpha_X$ for arbitrary $f$, one cannot
similarly conclude that $\alpha_X=0$.
}
that $\beta_X=0$ \cite{GratusTuckerObukhov2012}. Finally choosing 
$X=X_a$ and using the field equation (\ref{eq:dHJ}) together with
(\ref{AB_def_Dtau}) gives
\begin{equation}\label{eq:Dtauc}
D \tau _a
= F ( X_a , \, \cdot \, ) \wedge \, J 
\end{equation}
which is equivalent to (\ref{eq:divT}), using (\ref{eq:divTDtau}).
{}From this derivation one concludes that, for {\em any} diffeomorphism and
gauge invariant action 
\begin{equation*}
S[A,J,g]=\int_{\Man}\Lambda(A,J,g)
\end{equation*}
constructed from
\begin{equation}
\Lambda(A,J,g)
=
\Lambda^{\textup{f}}(F,g)
-
A\wedge J
\end{equation}
and $F=dA$ equation 
(\ref{eq:Dtauc}) is satisfied when
$J$ and $g$ are background fields with $\Lambda^{\textup{f}}(F,g)$ arbitrary.
However for consistency this requires that $J$ be a prescribed exact 3-form.
On a topologically trivial spacetime domain $\Man$ this is implied by $dJ=0$.

If the background metric is flat and inertial coordinates
are used, (\ref{eq:dHJ}) reduces to (\ref{eq:bpF}); 
in this case the components of the 
stress-energy-momentum tensor associated with
  (\ref{eq:tauc}) are
\begin{eqnarray}\label{eq:Tabflat}
4 \pi T_{cd} &=&
\frac{1}{4} F_{ab}F^{ab} \eta _{cd}-F_{db}F_{c}{}^{b}
+ \ell ^2 
\Big(
\partial ^b F_{bc}\partial ^a F_{ad}- 
\frac{1}{2} \partial ^aF_{ab} \partial _{\ell} F^{\ell b} \eta _{cd} \Big)
\\ &&
- \ell ^2 \Big( F_{cb} \partial ^a \partial _a F^b{}_d 
+ F_{db} \partial ^b \partial ^aF_{ac} 
+ F_{cb} \partial ^b \partial ^aF_{ad} 
+ F^{ab} \partial ^e \partial _e F_{ab} \eta _{cd}
\Big)
\nonumber
\end{eqnarray}
where we have used $dF=0$. In this flat metric with the field
equations (\ref{eq:bpF}), the stress-energy-momentum tensor
(\ref{eq:Tabflat}) satisfies $\partial ^c T_{cd} = F_{db}j^b$ which is
the flat-space coordinate version of (\ref{eq:divT}).  The
stress-energy-momentum tensor (\ref{eq:Tabflat}) coincides with
that written down by \citeasnoun{Podolsky1942} and subsequently
  quoted by \citeasnoun{Zayats2013}.
The {\em derivation} above shows that, in general, it coincides with
  the Hilbert stress-energy-momentum tensor derived from metric
  variations of the action (\ref{eq:SBAg}).

Finally we point out how different definitions of the stress-energy-momentum tensor for the Bopp-Podolsky
theory are responsible for the historic problems outlined in the
beginning of this appendix. In general (\ref{eq:Adef}) becomes
\begin{equation}\label{eq:tAhA}
\tilde{A} = - \ell ^2 \star G \, , \quad
\hat{A} = A - \ell ^2 \star G \, .
\end{equation}
and the Lagrangian (\ref{eq:SBAg}) can be rewritten as
\begin{equation}\label{BP_Lambda_equiv}
\Lambda
= \Lambda^1 +
\frac{\ell^2}{4\pi} \ d\big(\star d\star F\wedge\star F\big)
\end{equation}
where
\begin{equation}\label{BP_alt_Lag}
\Lambda^1 = 
\frac{1}{8\pi} d \hat{A} \wedge\star d\hat{A} 
- \frac{1}{8\pi} \left(d \tilde{A} \wedge\star d \tilde{A}  
+ \frac{1}{\ell^2} \tilde{A}\wedge\star \tilde{A}\right) - 
(\hat{A}-\tilde{A})\wedge J
\, .
\end{equation}
As $\Lambda$ and 
$\Lambda^1$ differ only by an exact form, they lead
to the same field equations and to the same \emph{Hilbert}
stress-energy-momentum tensor. Indeed, varying the action 
\begin{eqnarray}\label{eq:S1}
\fl
S^1[\hat{A},\tilde{A},g] &=& \int_{\Man} \Lambda^1 
\\\nonumber \fl &=& 
\int_{\Man} \bigg(\frac{1}{8\pi} d \hat{A} \wedge\star d\hat{A} 
- \frac{1}{8\pi} \left(d \tilde{A} \wedge\star d \tilde{A}  
+ \frac{1}{\ell^2} \tilde{A}\wedge\star \tilde{A}\right) - 
(\hat{A}-\tilde{A})\wedge J\bigg)
\end{eqnarray}
with respect to $\hat{A}$ and $\tilde{A}$ respectively
yields the field equations
\begin{equation}\label{eq:fieldtAhA}
d \star d \hat{A} = 4 \pi J \, , \quad
d \star d \tilde{A} - \frac{1}{\ell ^2} \star \tilde{A}
= 4 \pi J \, .
\end{equation}
These equations imply that $\tilde{A}$ necessarily satisfies the
Lorenz gauge condition, $d \star \tilde{A} =0$. 

The field equations (\ref{eq:fieldtAhA}) reduce to (\ref{eq:bphat}) and
  (\ref{eq:bptilde}) in flat spacetime using inertial
  coordinates and the Lorenz gauge condition on $\hat{A}$.

The $g-$orthonormal co-frame variation of $\Lambda^1$ gives
rise to the Hilbert stress forms 
\begin{equation}\label{eq:tauBopp}
\fl
8 \pi \tau _c = d \hat{A} \wedge i_c \star d \hat{A}
-i_c d \hat{A} \wedge \star d \hat{A} 
- d \tilde{A} \wedge i_c \star d \tilde{A}
+i_c d \tilde{A} \wedge \star d \tilde{A}
+ \frac{1}{\ell ^2} \Big( \tilde{A} \wedge i_c \star \tilde{A} 
+ i_c \tilde{A} \wedge \star \tilde{A} \Big) 
\, .
\end{equation}
\cRed{Substituting} from (\ref{eq:tAhA}), we see that 
these stress forms indeed coincide with (\ref{eq:tauc}).
On flat spacetime in inertial coordinates, (\ref{eq:tauBopp})
yields the stress-energy-momentum tensor
\begin{equation}\label{eq:TBopp}
\fl
4 \pi T_{cd} = 
\frac{1}{4} \hat{F}{}_{ab} \hat{F}{}^{ab} \eta _{cd}
- \hat{F}{}_{ac} \hat{F}{}^a{}_d
- \frac{1}{4} \tilde{F}{}_{ab} \tilde{F}{}^{ab} \eta _{cd}
+ \tilde{F}{}_{ac} \tilde{F}{}^a{}_d
+ \frac{1}{\ell ^2} \Big( \tilde{A}{}_c \tilde{A}{}_d
- \frac{1}{2} \tilde{A}{}_a \tilde{A}{}^a \eta _{cd} \Big)
\end{equation}
where $\hat{F}{}_{ab} = \partial _a \hat{A}{}_b
- \partial _b \hat{A}{}_a$ and $\tilde{F}{}_{ab} = 
\partial _a \tilde{A}{}_b - \partial _b \tilde{A}{}_a$.
This is the stress-energy-momentum tensor given by 
\citeasnoun{Bopp1940}, cf. again \citeasnoun{Zayats2013}. 

Furthermore the same stress forms (\ref{eq:tauBopp})
  arise using the canonical Belinfante-Rosenfeld symmetrization of the Noether
  current associated with the Lagrangian
  (\ref{BP_alt_Lag}). By contrast the canonical Belinfante-Rosenfeld procedure
  applied to the action (\ref{eq:Lag}) does not yield
  the stress-energy-momentum forms (\ref{eq:tauBopp}).
This illustrates the fact that two Lagrangians that differ by an exact form, 
while yielding the same Hilbert stress-energy-momentum tensors, 
in general, yield
different stress-energy-momentum tensors following the canonical
Belinfante-Rosenfeld procedure.


\References

\harvarditem{Benn \harvardand\ Tucker}{1988}{BennTucker1988}
Benn I \harvardand\ Tucker, R~W  1988 {\em An Introduction to Spinors and
  Geometry With Applications in Physics} (Bristol: Adam Hilger)

\harvarditem{Bopp}{1940}{Bopp1940}
Bopp F  1940 {\em Annalen der Physik} {\bf 430}~345--384

\harvarditem{Born \harvardand\ Infeld}{1934}{BornInfeld1934}
Born M \harvardand\ Infeld L  1934 {\em Proc. Roy. Soc. London} {\bf A
  144}~425--451

\harvarditem{Dereli et~al.}{2007}{DereliGratusTucker2007}
Dereli T, Gratus J \harvardand\ Tucker R  2007 {\em J. Phys. A} {\bf
  40},~5695--5715.

\harvarditem{DeWitt \harvardand\ Brehme}{1960}{DeWittBrehm1960}
DeWitt~B~S \harvardand\ Brehme~R~W 1960 {\em Ann. Phys.} {\bf
  9}(2)~220--259

\harvarditem{Dirac}{1938}{Dirac1938}
Dirac P~A~M  1938 {\em Proc. Roy. Soc. London} {\bf A 167}~148--169

\harvarditem{Ferris \harvardand\ Gratus}{2011}{FerrisGratus2011}
Ferris M~R \harvardand\ Gratus J  2011 {\em J. Math. Phys.} {\bf
  52}(9)~092902

\harvarditem{Frenkel}{1996}{Frenkel1996}
Frenkel J  1996 {\em Phys. Rev.} {\bf E 54}~5859--5862

\harvarditem{Frenkel \harvardand\ Santos}{1999}{FrenkelSantos1999}
Frenkel J \harvardand\ Santos R  1999 {\em Int. J. Modern Phys.} {\bf B
  13}~315--324

\harvarditem{Gratus et~al.}{2012}{GratusTuckerObukhov2012}
Gratus J, Tucker R~W \harvardand\ Obukhov Y~N  2012 {\em Ann. Phys. (NY)} {\bf
  327}~2560--2590

\harvarditem{Land{\'e} \harvardand\ Thomas}{1941}{LandeThomas1941}
Land{\'e} A \harvardand\ Thomas L~H  1941 {\em Phys. Rev.} {\bf 60}~514--523

\harvarditem{Norton}{2009}{Norton2009}
Norton A~H 2009 {\em Class. Quantum Grav. 26} {\bf 105009}

\harvarditem{Pais \harvardand\ Uhlenbeck}{1950}{PaisUhlenbeck1950}
Pais A \harvardand\ Uhlenbeck G~E  1950 {\em Phys. Rev.} {\bf 79}~145--169

\harvarditem{Pavlopoulos}{1967}{Pavlopoulos1967}
Pavlopoulos T~G  1967 {\em Phys. Rev.} {\bf 159}~1106--1110

\harvarditem{Podolsky}{1942}{Podolsky1942}
Podolsky B  1942 {\em Phys. Rev.} {\bf 62}~68--71

\harvarditem{Podolsky \harvardand\ Kikuchi}{1944}{PodolskyKikuchi1944}
Podolsky B \harvardand\ Kikuchi C  1944 {\em Phys. Rev.} {\bf 65}~228--235

\harvarditem{Podolsky \harvardand\ Kikuchi}{1945}{PodolskyKikuchi1945}
Podolsky B \harvardand\ Kikuchi C  1945 {\em Phys. Rev.} {\bf 67}~184--192

\harvarditem{Podolsky \harvardand\ Schwed}{1948}{PodolskySchwed1948}
Podolsky B \harvardand\ Schwed P  1948 {\em Rev. Modern Phys.} {\bf
  20}~40--50

\harvarditem{Poisson et~al.}{2011}{PoissonPoundVega2011}
Poisson E, Pound A \harvardand\ Vega I  2011 {\em Living Rev. Relativity} 
{\bf  14(7)}

\harvarditem{Rohrlich}{2007}{Rohrlich2007}
Rohrlich F  2007 {\em {Classical Charged Particles}} 
(Singapore: World Scientific)

\harvarditem{Spohn}{2007}{Spohn2007}
Spohn H  2007 {\em Dynamics of Charged Particles and their Radiation Field}
(Cambridge: Cambridge University Press)

\harvarditem{Straumann}{1984}{Straumann1984}
Straumann N  1984 {\em General Relativity and Relativistic Astrophysics}
(Berlin:  Springer)

\harvarditem{Zayats}{2014}{Zayats2013}
Zayats A~E  2014 {\em Ann. Phys. (NY)} {\bf 342}~11--20

\endrefs

\end{document}